\documentclass[reprint,aps,prx]{revtex4-1}

\usepackage[utf8]{inputenc}
\usepackage[svgnames]{xcolor}
\usepackage{xspace,graphicx}

\usepackage{microtype,enumitem}
\usepackage{dsfont,upgreek}
\usepackage{amsmath,amsthm,amssymb,mathtools}
\usepackage[colorlinks]{hyperref}

\usepackage{qcircuit}
\input{braket}

\usepackage{tikz}
\usetikzlibrary{arrows}
\usetikzlibrary{shapes.misc}

\usepackage[capitalise]{cleveref}
\crefname{lemma}{Lemma}{Lemmas}
\crefname{proposition}{Proposition}{Propositions}
\crefname{definition}{Definition}{Definitions}
\crefname{theorem}{Theorem}{Theorems}
\crefname{conjecture}{Conjecture}{Conjectures}
\crefname{corollary}{Corollary}{Corollaries}
\crefname{example}{Example}{Examples}
\crefname{section}{Section}{Sections}
\crefname{appendix}{Appendix}{Appendices}
\crefname{figure}{Fig.}{Figs.}
\crefname{equation}{Eq.}{Eqs.}
\crefname{table}{Table}{Tables}
\crefname{item}{Property}{Properties}
\crefname{remark}{Remark}{Remarks}

\newtheorem{theorem}{Theorem}

\newtheorem{corollary}[theorem]{Corollary}

\newtheorem{lemma}[theorem]{Lemma}

\newtheorem{remark}[theorem]{Remark}

\newcommand\field\mathds
\newcommand\op{\mathbf}
\newcommand\Hs{\mathcal H}
\renewcommand\1{\mathds{1}}

\newcommand\ii{\mathrm{i}}
\newcommand\ee{\mathrm{e}}

\DeclareMathOperator\BigO{O}

\DeclareMathOperator\poly{poly}
\DeclareMathOperator\lmin{\lambda_\mathrm{min}}
\DeclareMathOperator\spec{spec}
\DeclareMathOperator\sig{sig}
\DeclareMathOperator\spn{span}
\DeclareMathOperator\sign{sign}

\newcommand\bl{\raisebox{.1ex}{$\vartriangleright$}}
\newcommand\bd{\raisebox{-.125ex}{$\blacksquare$}}
\newcommand\hd{\raisebox{.1ex}{$\blacktriangleright$}}

\newcommand{\HTM}{\op H_{\mathrm{TM}}}
\newcommand{\Hel}{\op H^{(f)}}
\newcommand{\Htriv}{\op H_\mathrm{trivial}}
\newcommand{\Hcont}{\op H_\mathrm{dense}}
\newcommand{\Hguard}{\op H_\mathrm{guard}}

\newcommand{\ivar}{\upeta}

\begin{document}
\title{Undecidability of the Spectral Gap in One Dimension}
\author{Johannes Bausch}
\affiliation{CQIF, DAMTP, University of Cambridge, UK}
\author{Toby S. Cubitt}
\affiliation{Department of Computer Science, University College London, UK}
\author{Angelo Lucia}
\affiliation{QMATH, Department of Mathematical Sciences and NBIA, Niels Bohr Institute, University of Copenhagen, 2100 Copenhagen, DK}
\affiliation{Walter Burke Institute for Theoretical Physics and Institute for Quantum Information \& Matter,\protect\\ California Institute of Technology, Pasadena, CA 91125, USA}
\author{David~Perez-Garcia}
\affiliation{Dept.\ An\'alisis y Matem\'atica Aplicada,%
Universidad Complutense de Madrid, 28040 Madrid, Spain}
\affiliation{Instituto de Ciencias Matem\'aticas, 28049 Madrid, Spain}

    \begin{abstract}
          The spectral gap problem---determining whether the energy spectrum of a system has an energy gap above ground state, or if there is a continuous range of low-energy excitations---pervades quantum many-body physics.
          Recently, this important problem was shown to be undecidable for quantum spin systems in two (or more) spatial dimensions: there exists no algorithm that determines in general whether a system is gapped or gapless, a result which has many unexpected consequences for the physics of such systems.
          However, there are many indications that one dimensional spin systems are simpler than their higher-dimensional counterparts: for example, they cannot have thermal phase transitions or topological order, and there exist highly-effective numerical algorithms such as DMRG---and even provably polynomial-time ones---for gapped 1D systems, exploiting the fact that such systems obey an entropy area-law.
          Furthermore, the spectral gap undecidability construction crucially relied on aperiodic tilings, which are not possible in 1D.\\
          So does the spectral gap problem become decidable in 1D?
          In this paper we prove this is not the case, by constructing a family of 1D spin chains with translationally-invariant nearest neighbour interactions for which no algorithm can determine the presence of a spectral gap.
          This not only proves that the spectral gap of 1D systems is just as intractable as in higher dimensions, but also predicts the existence of qualitatively new types of complex physics in 1D spin chains.
          In particular, it implies there are 1D systems with constant spectral gap and non-degenerate classical ground state for all systems sizes up to an uncomputably large size, whereupon they switch to a gapless behaviour with dense spectrum.
    \end{abstract}
    \maketitle

\section{Introduction}
One-dimensional spin chains are an important and widely-studied class of quantum many-body systems.
The quantum Ising model, for example, is a classic model of magnetism; the 1D Ising model with transverse fields is the textbook example of a quantum phase transition.
It is also one of a handful of quantum many-body systems which can be completely solved analytically. Indeed, most known exactly solvable quantum many-body models are in 1D~\cite{Lieb1961a,Affleck1987,Franchini2017}.
Even for 1D systems that are not exactly solvable, the density matrix renormalisation group (DMRG) algorithm~\cite{White1992} works extremely well in practice, and recent results have even yielded provably efficient classical algorithms for all 1D gapped systems~\cite{Landau2013}.

While it is known that approximating a 1D quantum system's ground state energy to inverse polynomial precision is in general QMA hard~\cite{Kitaev2002,Aharonov2009}---even with translationally-invariant nearest neighbour interactions~\cite{gottesman2009quantum,Bausch2016}---currently there are no examples of gapped QMA-hard Hamiltonians and and there are indications~\cite{1810.06528} that gaplessness is required in order to have a hard to compute ground state energy.

There are several other indications that ground states of (finite) gapped 1D systems are qualitatively simpler than in higher dimensions.
They obey an entanglement area-law, hence have an efficient classical descriptions in terms of matrix product states~\cite{Hastings2007a,Arad2013}.
Furthermore, thermal phase transitions~\cite{Imry1974} and topological order~\cite{Verstraete2005} are both ruled out for 1D quantum systems. For classical 1D systems, satisfiability and tiling problems become tractable.
For the simplest class of spin chains---qubit chains with translationally invariant nearest-neighbour interactions---the spectral gap problem has been completely solved when the system is frustration-free~\cite{Bravyi2015}.

Contrast this with the situation in 2D and higher, where even simple theoretical models such as the 2D Fermi-Hubbard model (believed to underlie high-temperature superconductivity) cannot be reliably solved numerically even for moderately large system sizes~\cite{Staar2013,Mazurenko2017};
the entropy area-law remains an unproven conjecture~\cite{Ge2016};
and the spectral gap problem---i.e.\ the question of existence of a spectral gap above the ground state in the thermodynamic limit---is undecidable~\cite{Cubitt2015,Cubitt2015_long}.
This latter result holds under the assumption that either the ground state is non-degenerate with a constant spectral gap above it in the gapped case, or that the entire spectrum is continuous in the gapless case: therefore the undecidability of the problem of distinguishing the two cases is not due to the presence of ambiguous cases (for example cases where low-excited states collapse onto the groundstate in the limit).
For classical systems, satisfiability and tiling problems are NP-hard~\cite{Cook1971} and undecidable~\cite{Berger1966} (respectively) in two dimensions and higher.

Despite these indications that one-dimensional systems appear qualitatively easier to analyse than their higher-dimensional counterparts, we show in this paper that the spectral gap problem is undecidable, even in 1D.
The many-body quantum systems we consider in this work are one-dimensional spin chains, i.e.\ with a Hilbert space $(\field C^d)^{\otimes N}$, where $d$ is the local physical dimension, and $N$ the length of the chain.
The spins are coupled by translationally-invariant local interactions: a nearest-neighbour term $\op h^{(2)}$, which is a $d^2\times d^2$ Hermitian matrix, and a $d\times d$-sized local term $\op h^{(1)}$ which is also Hermitian.
Both $\op h^{(1)}$ and $\op h^{(2)}$ are independent of the system size $N$.
The overall Hamiltonian $\op H_N$ will be a sum of the local terms:
\begin{equation}\label{eq:H_N}
  \op H_N = \sum_{i=1}^{N-1} \op h^{(2)}_{i,i+1} + \sum_{i=1}^N \op h^{(1)}_i.
\end{equation}
(Following standard notation, subscripts indicate the spin(s) on which the operator acts non-trivially, with the operator implicitly extended to the whole chain by tensoring with $\1$ on all other spins.)
More precisely, $\op h^{(1)}$ and $\op h^{(2)}$ define a sequence of Hamiltonians $\{\op H_N\}$ on increasing chain lengths.
The thermodynamic limit will be taken by letting $N$ grow to infinity.

In order to be completely unambiguous about what we mean by the two terms \emph{gapped} and \emph{gapless}, we use a very strong definition. For $\{\op H_N\}$ to be gapless, we require that there exists a finite interval of size $c>0$ above its ground state energy $E_0(N)$ such that the spectrum of $\op H_N$ becomes dense therein as $N$ goes to infinity, in the sense that any value in the interval $[E_0(N), E_0(N)+c]$ is arbitrarily well approximated by a $N$-dependent sequence of eigenvalues of $\op H_N$.
In contrast, $\{\op H_N\}$ is gapped if there exists  $\gamma>0$ such that for all $N\in \field{N}$, $\op H_N$ have a non-degenerate ground state and a spectral gap $\Delta(\op H_N)>\gamma$ where $\Delta(\op H_N)$ is the difference in energy between the (unique) ground state and the first excited state~\footnote{Note that gapped is not defined as the negation of gapless; there are systems that fall into neither class. The reason for choosing such strong definitions is to deliberately avoid ambiguous cases (such as systems with degenerate ground states). Our constructions will allow us to use these strong definitions, because we are able to guarantee that each instance falls into one of the two classes.} (see \cref{fig:gapped-vs-dense}).

\section{Main result}
Our main result is a construction of a nearest-neighbour coupling $\op h^{(2)}(\ivar)$ and a single site term $\op h^{(1)}(\ivar)$, parametrized by an integer $\ivar$, with the guarantee that each of the corresponding Hamiltonians $\{ \op H_N(\ivar) \}$, defined via \eqref{eq:H_N}, is either gapped or gapless according to the definitions given above. For this particular class of Hamiltonians, we show that determining which $\ivar$ correspond to gapped instances and which $\ivar$ correspond to gapless instances is as hard as determining whether a given Turing machine halts, a problem known as the Halting problem. Since the latter problem is undecidable~\cite{Turing1937a}, this immediately implies that the question of existence of a spectral gap is also undecidable for 1D Hamiltonians, both algorithmically, as well as in the axiomatic sense of Gödel~\cite{Godel1931}. 

The construction of the interactions $\op h^{(2)}(\ivar)$ and $\op h^{(1)}(\ivar)$ is based on an embedding of a fixed universal Turing machine (UTM), in such a way that the spectral gap problem for $\{\op H_N(\ivar)\}$ encodes the behaviour of the UTM when given $\ivar$ as an input: if the UTM halts on input $\ivar$, then $\{\op H_N(\ivar)\}$ will be gapless, while if the UTM does not halt on input $\ivar$, it will be gapped with spectral gap uniform in $\ivar$.

Moreover, we can show that $\op h^{(2)}(\ivar)$ and $\op h^{(1)}(\ivar)$ can be choosen to be small quantum perturbations around of a classical interaction (i.e. diagonal in the computational basis), and that their depedence on $\ivar$ is only due to some numerical factors. We present this explicit form, toghether with a summary of the above discussion, in the following theorem.

\begin{theorem}\label{th:main}
  Fix a universal Turing machine (UTM).
  There exist (explicitly constructible) nearest-neighbor interactions $\op h^{(2)}(\ivar)$ and a local term $\op h^{(1)}(\ivar)$, parametrized by an integer $\ivar$, such that $\|\op h^{(1)}(\ivar)\|\leq 2$, $\|\op h^{(2)}(\ivar) \| \leq 1$ and the family of Hamiltonians $\{\op H_N(\ivar)\}$ defined on a spin chain with $N$ sites and local dimension $d$ by
    \[  \op H_N(\ivar) = \sum_{i=1}^{N-1} \op h^{(2)}_{i,i+1}(\ivar) + \sum_{i=1}^N \op h^{(1)}_i(\ivar), \]
    satisfies the following:
    \begin{enumerate}[leftmargin=*]
    \item if the UTM halts on input $\ivar$, then $\{\op H_N(\ivar)\}$ is gapless.
    \item if the UTM does not halt on input $\ivar$, then  $\{\op H_N(\ivar)\}$ is gapped.
        Moreover the spectral gap $\Delta(\op H_N(\ivar))\ge 1$ for all $N\in \field{N}$.
    \end{enumerate}
    
  The interactions  $\op h^{(2)}(\ivar)$ and $\op h^{(1)}(\ivar)$ can be chosen to be of the form 
    \begin{alignat}{2}
      \op h^{(1)}(\ivar) & = \op a + \beta(&&2^{-2|\ivar|}\op a' + \op a'' ),\label{eq:mainthm-single-site-interaction}\\
      \op h^{(2)}(\ivar) & = \op b + \beta(&&2^{-2|\ivar|}\op b' + \op b''+  \notag \\
      & && \ee^{\ii \pi\phi(\ivar)} \op b''' + \ee^{-\ii \pi\phi(\ivar)}\op b'''^\dagger+  \notag \\
      & && \ee^{\ii \pi2^{-2|\ivar|}} \op b'''' + \ee^{-\ii \pi2^{-2|\ivar|}}\op b''''^\dagger). \label{eq:mainthm-two-sites-interaction}
    \end{alignat}
  where $0<\beta\le 1$ is any rational number (which can be chosen arbitrarily small),
   $|\ivar|$ denotes the number of digits in the binary expansion $\ivar=\ivar_1\ivar_2\ldots \ivar_{|\ivar|}$, $\phi(\ivar)$ denotes its binary fraction with interleaved $1$s, i.e. $\phi(\ivar)=0.\ivar_11\ivar_21\ldots \ivar_{|\ivar|-1}1\ivar_{|\ivar|}$, 
   and $\op a, \op a', \op a''$ are $d\times d$ matrices and $\op b,\op b',\op b'', \op b''', \op b''''$ are $d^2 \times d^2$ matrices with the following properties:
  \begin{enumerate}
  \item $\op a$ and $\op b$ are diagonal with entries in $\field Z$, i.e.\ they correspond to a purely classical spin coupling.
  \item $\op a',\op a''$, $\op b', \op b''$ are Hermitian with entries in $\field Q[\sqrt 2]$, i.e. they are of the form $x+y\sqrt{2}$ with $x$ and $y$ being rational numbers.
  \item $\op b''', \op b''''$ have entries in $\field Q$.
\end{enumerate}
Since the matrices constructed have entries in $\field Q[\sqrt 2]$, they can be specified by a finite description, which toghether with the binary expansion of $\ivar$ completely determines the interactions  $\op h^{(2)}(\ivar)$ and $\op h^{(1)}(\ivar)$.
\end{theorem}
As in the 2D case, we emphasize that, since $\beta$ can be an arbitrarily small parameter, the theorem proves that even an arbitrarily small perturbation of a classical Hamiltonian can have an undecidable spectral gap in the thermodynamic limit. This also shows that even for classical Hamiltonians, the gapped phase is not stable in general and is susceptible to arbitrarily small perturbations.

There have been many previous results over the years relating undecidability to classical and quantum physics~\cite{Komar64,Anderson72,Pourel81,Fredkin82,Domany84,Omohundro84,Gu-Nielsen,Wang,Berger1966,Kanter90,Moore90,Bennett90,Eisert12,Wolf11,Morton12,Kliesch14,Delascuevas16,VandenNest08,Elkouss16,Bendersky16,Slofstra16,Lloyd,Lloyd94,Lloyd16,Cubitt16-Comment}.
We refer to the introduction of~\cite{Cubitt2015_long} for a detailed historical account of these previous results.

\begin{figure}
\centering
\includegraphics[width=8cm]{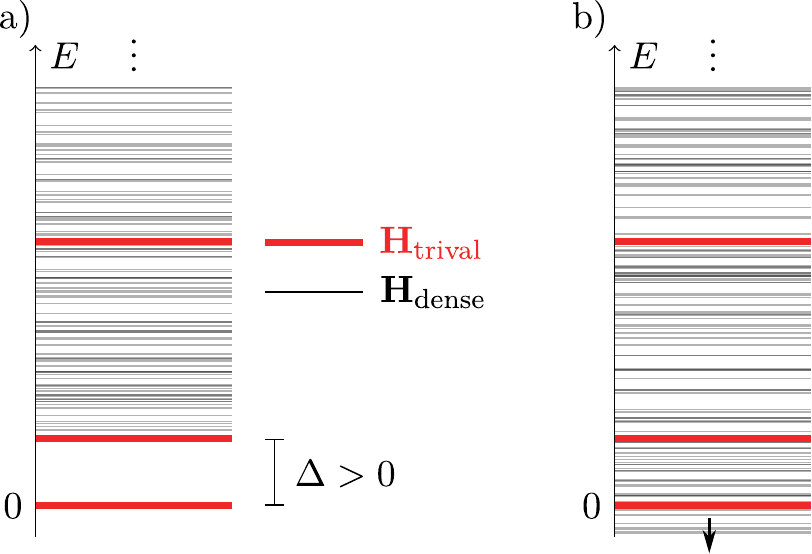}
\caption{Competing spectra of gapless versus gapped phase for the Hamiltonian
  $\op H=(\op H_C + \Hcont)\oplus 0 + 0\oplus \Htriv$.
  a) The system is gapped with $\Delta>0$ and unique product ground state. The thermodynamic limit is in a gapped phase.
  b) If and only if the encoded universal Turing machine halts, there exists a critical threshold system size after which the dense spectrum of $\op H_C+\Hcont$ is pulled towards $-\infty$ as the system size increases, covering up
  the gap in the spectrum of $\op H_\mathrm{trivial}$. The thermodynamic limit is in a gapless phase.}
\label{fig:gapped-vs-dense}
\end{figure}

So where is the difficulty in extending the two-dimensional result of \textcite{Cubitt2015} to one-dimensional systems?
One of the key ingredients in the 2D construction is a classical aperiodic tiling.
The particular tiling used in~\cite{Cubitt2015_long}, due to Robinson~\cite{Robinson1971}, exhibits a fractal structure, i.e.\ a fixed density of structures at all length scales.
This ingredient is crucial if one were to directly translate the original undecidability result to a one-dimensional system.

A Wang tile set~\footnote{We restrict to Wang tiles for simplicity. Slightly more general tiling rules are also possible, such as requiring complementary colours on abutting sides, but the same argument applies.} consists of a finite set of different types of square tiles, each tile type having one colour assigned to each of its four sides~\cite{Wang}. Translated to Hamiltonians, the computational basis state at each site indicates which tile is placed there; the interactions of the corresponding tiling Hamiltonian are diagonal projectors in the computational basis; each of these projectors constrains neighbouring sites to be in states that correspond to a matching tile configuration (i.e.\ where two tiles can only be placed next to each other if the colours of the abutting sides match).
A constant local dimension implies we can only have a constant number of tiles and thus of colors.
But in 1D, as soon as any tile occurs a second time along the chain, the entire pattern that followed that tile previously can repeat indefinitely. (Conversely, just as in 2D, if any finite segment cannot be tiled, then neither can the infinite chain.) Thus the Tiling problem in 1D is known to be decidable, even by a simple algorithm.

For this reason, an underlying tile set like the Robinson tiles used in 2D---with patterns of all length scales---is impossible in 1D, under the physical constraint of retaining a finite local dimension. 

Quantum mechanics can in principle circumvent this constraint, since entanglement can introduce long-range correlations, even in unfrustrated qudit chains \cite{Movassagh2010d}.
Yet even though it is known that one can obtain correlations between far-away sites that decay only polynomially, the resulting Hamiltonians are gapless \cite{Hastings2006,Movassagh2014}.

The key new idea is a 1D construction, which we denote the Marker Hamiltonian, that creates---within the system's ground state---a \emph{periodic} partition of the spin chain into segments, but whose length and period are related to the halting time of a Turing machine. This subtle interplay between the dynamics of a Turing machine, the periodic quantum ground state structure and the energy spectrum, plays the role of the classical aperiodic tilings of the 2D construction.

The paper is organized as follows. In Section \ref{sec:outline} we present a summary of the construction and how it differs from the 2D one.
In Section \ref{sec:marker} we detail the construction of the Marker Hamiltonian.
In Section \ref{sec:qpe} we present the modifications that are required to the encoding of UTM into Hamiltonian interactions due to this modified set-up. These two components will be combined in Section \ref{sec:combining}. The main result on the undecidability of the spectral gap will be proven in Section \ref{sec:undecidability}. Finally we present some extensions to our result in Section \ref{sec:extensions}.

\section{Outline of the construction}\label{sec:outline}
Let us now give an outline of how we circumvent the problems in extending the 2D construction to 1D chains, and present an overview of the different elements which will be required to construct $\op h^{(2)}(\ivar)$ and $\op h^{(1)}(\ivar)$.

We start by presenting some background on Turing machines and the Halting problem.
\subsection{Turing machines}
A (classical) Turing machine is a simple model of computation consisting of an infinite ``tape'' divided into cells, and a ``head'' which steps left or right along the tape.
The machine is always in one of a finite number of possible ``internal states'' $\{q_i\}_{i=1}^Q$.
There is one special internal state, denoted $q_f$, which tells the machine to halt when it enters this state.
Each cell can have one ``symbol'' written in it, from a finite set of possible symbols $\{\sigma_{i=1}^\Sigma\}$.
A finite table of ``transition rules'' determine how the machine should behave for each possible combination of symbol and internal state.
At each time step, the machine reads the symbol in the cell currently under the head and looks up this symbol and the current internal state in the transition rule table.
The transition rule specifies a symbol to overwrite in the current cell, a new internal state to transition to, and whether to move the head left or right one step along the tape.
The ``input'' to a Turing machine is whatever symbols are initially written on the tape, and the ``output'' is whatever is left written on the tape when it halts.

Despite its apparent simplicity, Turing machines can carry out any computation that it is possible to perform.
Indeed, Turing constructed a universal Turing machine (UTM): a \emph{single} set of transition rules that can perform any desired computation, determined solely by the input.
Given an input $\ivar$ to a universal Turing machine $M$, the Halting Problem asks whether $M$ halts on input $\ivar$.

\subsection{Encoding of the Halting problem}\label{sec:hist-state-intro}
We want to construct a Hamiltonian whose spectral gap encodes the Halting Problem.
More precisely, starting from a fixed UTM $M$, we want to construct the interactions $\op h^{(2)}(\ivar)$ and $\op h^{(1)}(\ivar)$  which define a 1D, translationally invariant, nearest-neighbour, spin chain Hamiltonian $\op H_N(\ivar) = \op H_N(M,\ivar)$ on the Hilbert space $\mathcal H=(\field C^d)^{\otimes N}$, such that $\op H_N(\ivar)$ is gapped in the limit $N\rightarrow\infty$ if $M$ halts on input $\ivar$, and gapless otherwise.

In the earlier 2D construction \cite{Cubitt2015_long}, this was accomplished by combining a trivial gapped Hamiltonian with one that has a dense spectrum (and thus gapless). The combined Hamiltonian has the property that the ground state energy is the smallest of the two.
The dense Hamiltonian is then modified such that, if $M$ halts on input $\ivar$, its lowest eigenvalue is pushed up by a large enough constant, revealing the gap present due to the trivial Hamiltonian. In a non-Halting instance the dense spectrum Hamiltonian has the lowest ground state energy and therefore the combined Hamiltonian remains gapless.

In order to modify the dense Hamiltonian in this fashion, we have to construct a Hamiltonian whose ground state energy is dependent on the outcome of a (quantum) computation. This is possible thanks to Feynman and Kitaev's history state construction, used ubiquitously throughout quantum complexity proofs~\cite{Feynman1986,Kitaev2002,Oliveira2008,Aharonov2009,Kempe2006,gottesman2009quantum,Breuckmann2013,Bausch2016,Bausch2017}.
In brief, this construction allows one to take a circuit $C$ with gates $\op U_1,\ldots,\op U_T$ acting on $m$ qubits, and embed it into a Hamiltonian on $n = m + \poly\log T$ qubits, such that the ground state is a superposition over histories of the computation, i.e.\ a state of the form $\ket\Psi\propto\sum_{t=0}^T\ket t\ket{\psi_t}$.
Every ``snapshot'' of the computation $\ket{\psi_t}$ is entangled with a so-called clock register $\ket t$.
For $T$ computational steps, one can implement such a clock with a local Hamiltonian using $\poly\log T$ qubits.
The state $\ket{\psi_0}$ is thus input to the circuit, and $\ket{\psi_t}=\op U_t\cdots\op U_1\ket{\psi_0}$ is the state of the circuit after $t$ gates.
A later construction due to \textcite{gottesman2009quantum} similarly encodes the evolution of a quantum Turing machine, instead of a quantum circuit.
As the transition rules of a Turing machine do not depend on the head location, a benefit of encoding Turing machines rather than circuits is that the resulting Hamiltonians are naturally translationally invariant.

By adding a projector to ``penalize'' a subset of the possible outcomes of the computation, as encoded in $\ket T\ket{\psi_T}$, the ground state in these cases is pushed up in energy by $\Theta(T^{-2})$.
We denote this circuit Hamiltonian with penalties with $\op H_C=\op H_C(M,\ivar)$---as it will be the only term dependent on the free parameter $\ivar$ and our chosen Turing machine $M$, set up such that $\ivar$ will serve as input to $M$---and the Hilbert space it acts on with $\mathcal H_C$.
The energy shift in $\op H_C$'s ground state can be exploited by combining this circuit Hamiltonian with three more Hamiltonians: $\Hcont$ with a non-negative and asymptotically dense spectrum on a Hilbert space $\mathcal H_\text{dense}$, and $\Htriv$ with a trivial zero energy ground state and gap $\ge 1$, on Hilbert space $\mathcal H_\text{trivial}$.
Then
$$
\op H_N(M,\ivar) \coloneqq (\op H_C(M,\ivar)+\Hcont)\oplus 0+0\oplus\Htriv + \Hguard
$$
is defined on the overall Hilbert space
$$
\mathcal H \coloneqq (\mathcal H_C \otimes \mathcal H_\text{dense}) \oplus \mathcal H_\text{trivial}.
$$
In order to ensure that the low-energy spectrum of $\op H_N$ is determined \emph{either} completely by $\Htriv$ \emph{or} by the sum $\op H_C+\Hcont$, we have added another local Hamiltonian $\Hguard$ acting on $\mathcal H$ with Ising-type couplings that penalize states with ``mixed'' support (explicitly spelled out in \cref{th:undecidability-1-gap}).

If the computation output in $\op H_C$ is penalized, the dense spectrum is pushed up, which in turn unveils the constant spectral gap of some trivial Hamiltonian $\Htriv$, as shown in  \cref{fig:gapped-vs-dense}.

Yet even though we can easily penalize an embedded Turing machine reaching a halting state in this way (i.e.\ by adding a penalty term for the head being in any terminating state $q_f$), a history state Hamiltonian is insufficient for the undecidability proof.
i) The energy penalty decreases as the embedded computation becomes longer~\cite{Bausch2016a}.
However, we require a constant energy penalty density across the spin chain.
ii) If we try to circumvent this problem by subdividing the tape to spawn multiple copies of the
Turing machine, we need to know the space required beforehand in which the computation halts, if it
halts---which is also undecidable.

\subsection{Amplifying the energy penalty}
\citeauthor{Cubitt2015_long} circumvent this problem by spawning a fixed \emph{density} of computations across an underlying Robinson lattice.
Like this, within every area $A$, the halting case obtains an energy penalty $\propto A$---the ground state energy density therefore differs by a constant for the Halting and non-Halting cases, allowing the ground state energy to diverge in the Halting case, which uncovers the spectral gap.
The fractal properties of the Robinson tiling further ensure that that every possible tape length appears with a non-zero density in the large system size limit, so knowledge of the Turing machine's required runtime space is unnecessary.

\begin{figure}%
	\centering
	\includegraphics[width=6.5cm]{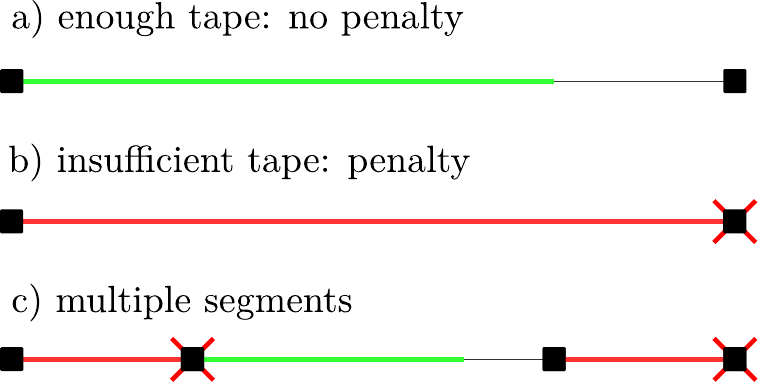}
	\caption{1D Robinson tiling analogue, the Marker Hamiltonian: penalty between halting and not halting for the TM is \emph{flipped}, i.e.\ we penalize \emph{not} halting (the TM head moves past the available tape, or equivalently the clock driving the TM runs out of space---see \cref{rem:loop}).
	a) If the tape---delimited by a black segment marker---is long enough for the TM to terminate, there is no penalty.
	b) If the tape is too short, a penalty will be inflicted due to the head running into the right segment marker.
	c) Mixed-length segments, each delimited with a segment marker.
	Those segments for which there is insufficient tape space pick up a penalty due to halting.
	The final construction introduces a small bonus for each segment, which shrinks the longer the segment is, and which is always smaller (in modulus) than the penalty that could be inflicted on the TM running on the available tape.
	In the halting case, this results in the lowest energy configuration being evenly-spaced segments with just enough tape for the TM to halt.
	In the non-halting case, a single segment is most favourable.
    }
	\label{fig:1d}
\end{figure}

We replace the fractal Robinson tiling with a 2-local ``marker'' Hamiltonian $\op H'$ on $(\field C^c)^{\otimes N}$, where
the markers---a special spin state $\ket\bd$---bound sections of tape used for the Turing machine.
$\op H'$ is diagonal with respect to boundary markers---i.e.\ $\op H'$ commutes with $\ketbra{\bd}{\bd}$.
Thus any eigenstate $\ket\psi$ has a well-defined \emph{signature} with respect to these boundaries, where the signature $\sig\ket\psi$ is defined as the binary string with 1's where boundaries are located, and 0's everywhere else.
We construct $\op H'$ in such a way that two consecutive markers bounding a segment will introduce an energy bonus that falls off quickly as the length of the segment increases: e.g.\ any eigenstate $\ket\psi$ with a signature
\[
	\sig\ket\psi=(\ldots,0,\underbrace{1,0,\ldots,0,1}_\text{length $w$},0,\ldots)
\]
will pick up a bonus of $\exp(-p(w))$ for some fixed polynomial $p$.
This bonus will be strictly smaller in magnitude than any potential penalty obtained from a computation running on the same segment of length $w$, i.e.\ when the TM head runs out of tape (see \cref{fig:1d}).

\subsection{Quantum phase estimation}
To the marker Hamiltonian, we add a history state Hamiltonian $\op H_\text{prop}(\phi,M)$.
Here
\begin{equation}\label{eq:phi-intro}
\phi=\phi(\ivar)=0.\ivar_11\ivar_21\ldots \ivar_{|\ivar|}00\ldots
\end{equation}
encodes an input parameter $\ivar\in\field N$ with $|\ivar|$ binary digits as binary fraction, where the digits of $\ivar$ are  interleaved by $1$s. The second parameter $M$ is a classical universal Turing machine.
We construct $\op H_\text{prop}$ to encode the following computation:
\begin{enumerate}
\item A Quantum Turing machine performs quantum phase estimation (QPE) on a single-qubit unitary that encodes the input $\phi$.
\item The classical universal TM $M$ uses the binary expansion of $\phi$ as input and performs a computation on it.
\end{enumerate}
Up to a slight modification for 1.\ which we will explain later, this is the same Turing machine construction as in~\cite[sec.~6]{Cubitt2015_long}.
The Hamiltonian $\op H_\text{prop}$ is set up to spawn one instance of the computation per segment, and we penalize the TM $M$ running out of available tape up to the next boundary marker with some local terms; as before we denote the resulting local Hamiltonian with $\op H_C$.
We finally add a trivial Hamiltonian $\Htriv$ with ground state energy $-1$ and constant spectral gap.
The overall Hamiltonian is then
\begin{align*}
    \op H_N=&\beta(\mu\op H'+\op H_C(M,\ivar)+\Hcont)\oplus 0\\
                        &+ 0\oplus\Htriv + \Hguard,
\end{align*}
where $\mu=2^{-|\phi|}$ is a small constant, defined for $\phi(\ivar)$ as given in \cref{eq:phi-intro} with $|\phi|=2|\ivar|$.
$\beta>0$ can be chosen arbitrarily small.

\begin{figure*}[t!]
	\centering
	\hspace*{-2cm}\includegraphics[]{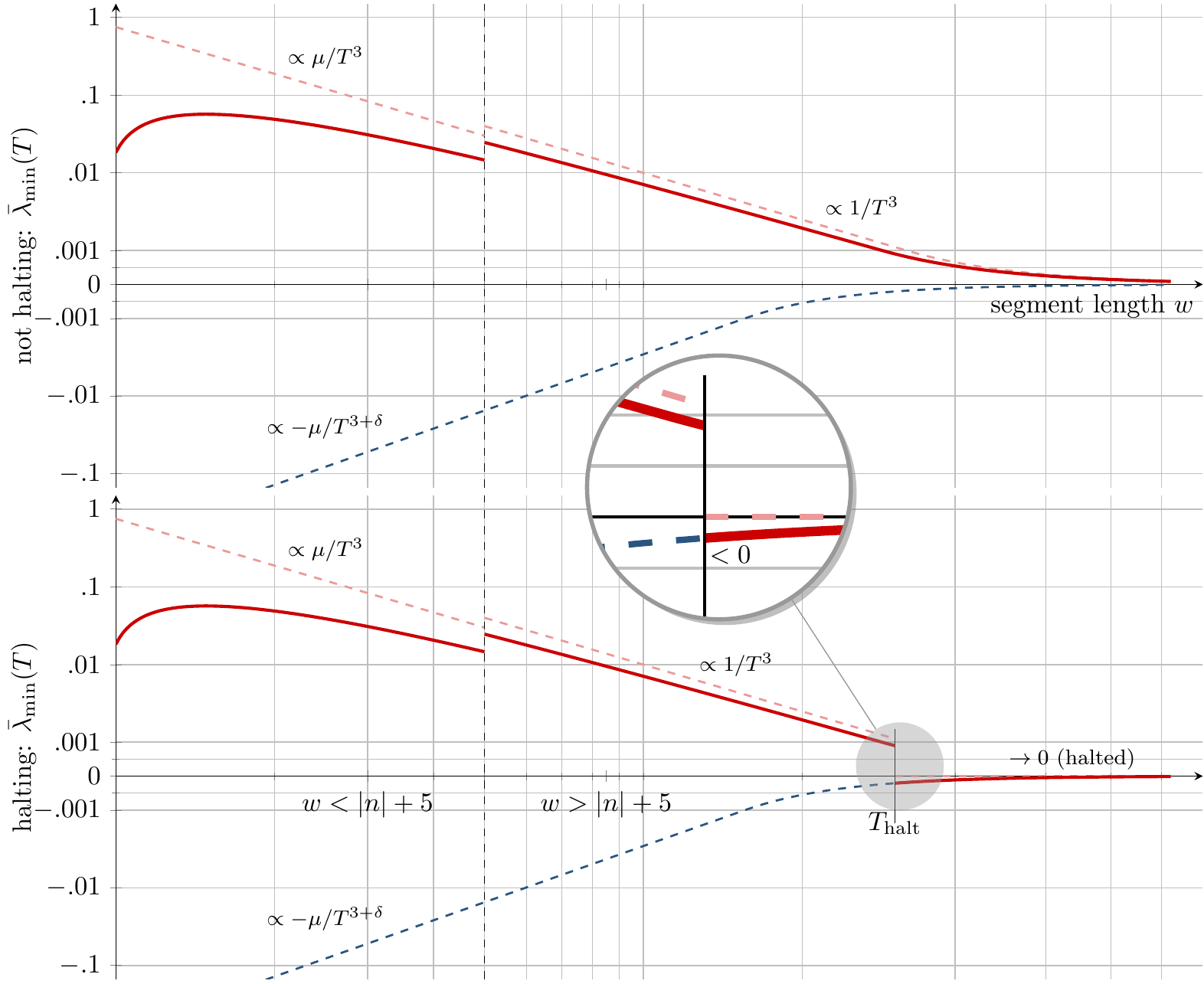}
    \caption{Energy contribution $\bar\lambda_\text{min}(T)$ from a single segment of length $w$ of the marker and TM Hamiltonian $\mu\op H'+\op H_C$ shown in red, where $T$ is the runtime of the encoded computation, bounded either by the segment length or by the halting time of the TM.
    The prefactor $\mu=2^{-2|\ivar|}$ is a small constant to compensate for the fact that on too short segments the phase estimation truncates the output, which we can only penalize with strength $\Omega(\mu/T^3)$. The dashed red line is the contribution of $\op H_C$, i.e.\ the energy penalty inflicted in case of the Turing machine running out of space. The dashed blue line is the bonus from $\mu\op H'$.}
    \label{fig:pen-and-bonux}
\end{figure*}

We will now explain how our construction differers during the QPE step.
QPE can be performed exactly when there is sufficient tape~\cite{Nielsen2010}.
In case there is insufficient space for the full binary expansion of the input parameter $\phi$, the output is truncated, and the resulting output state is not necessarily a product state in the computational basis anymore.

As in the 2D model, we have to allow for the possibility that the QPE truncates $\phi$, possibly resulting in the universal TM dovetailed to the QPE switching its behavior to halting.
In the 2D construction of~\cite{Cubitt2015_long}, one could circumvent this by simply subtracting off the energy contribution from truncated phase-estimation outputs; yet we cannot use this mechanism in our result, since it is not possible in the 1D construction, since we cannot \`a priori know the length of the segments on which the Turing machine runs.
Instead, we augment the QPE algorithm by a short program which verifies that the expansion has been performed in full, and otherwise inflicts a large enough energy penalty to offset the case that the UTM now potentially halts on the perturbed QPE output.

To this end, we make use of the specific encoding of $\phi$: the interleaved $1$s are flags indicating how many digits to expand.
Like this, before the inverse quantum Fourier transform, we know that the least-significant qubit is exactly in state $\ket+$ if the expansion was completed, and has overlap at least $\mu=2^{-|\phi|}$ with $\ket-$ otherwise.
By adding a penalty term to the Hamiltonian for said digit in state $\ket-$, we can penalize those segments with insufficient tape for a full expansion of the input, independently of whether the universal TM then halts or not on a faulty input.
This result manifests as a kink of the lower energy bound for a too-short segment of length $w$ in \cref{fig:pen-and-bonux}.
Yet since the marker Hamiltonian $\op H'$ is attenuated by $\mu$ as well, the energy remains nonnegative throughout for these segments.
Therefore, the only segments left to be analyzed are those for which the input can be assumed un-truncated.

\subsection{Ground state energy analysis}
When there is enough space for the QPE to be permormed, there are two possibilities for the ground state energy of $\op H_N$.
In case $M(\phi)$ does \emph{not} halt, any instance of the TM running on any tape length will
run out of tape space, incurring the penalty explained in \cref{fig:1d}. This halting
penalty will always dominate the bonus coming from the segment length, and we show the ground state energy to be $\lambda_\text{min}(\op H_N)\ge 0$.
In case the TM \emph{does} halt, there will be minimal segment length $w_\text{halt}$ above which
segments will not pick up the penalty from exhausting the tape.
Since the bonus given by the Marker Hamiltonian is decreasing with increasing segment length, the optimal energy configuration will therefore be achieved by partitioning the whole chain into segments of length $w_\text{halt}$, each of which picks up a tiny---but finite---negative energy contribution.
We prove $\lambda_\text{min}(\op H_N)<-\lfloor N / w_\text{halt} \rfloor \Omega(1/T_\text{halt}^3)$ in that case, where $T_\text{halt}$ is the number of computation steps till halting.
As the system size $N$ increases, the ground state energy will therefore diverge to $-\infty$.

The the claims of \cref{th:main} will then follow by combining the construction outlined with a trivial, gapped Hamiltonan and a dense spectrum, gapless Hamiltonian. The dense spectrum Hamiltonian will be modified to have ground state energy determined by the outcome of the computation of the QTM running on the tape segments defined by the Marker Hamitonian, so that the low-energy part of the spectum of the combined Hamiltonian will be gapped or gapless depending on whether the UTM halts or not.

\section{Marker Tiling}\label{sec:marker}
In this section we will give an explicit construction of the Marker Hamiltonian.
\subsection{Concept}
\DeclareDocumentCommand{\robinson}{ m m m m m }{
	\begin{tikzpicture}[
	scale=.55,
	x=3mm,
	y=3mm,
	nohalt/.style={green},
	halt/.style={red}
	]
	\clip (-.1,11.8) rectangle (47, 47);
	\foreach \s/\o/\n/\opacity/\tm in {%
		1/0/23/.4/#1, 4/1.5mm/5/.6/#2, 16/7.5mm/1/.8/#3, 64/-64.5mm/0/1./#4%
	} {
		\begin{scope}[
		xshift=-\o,
		yshift=-\o,
		opacity=\opacity
		]
		\begin{scope}[
		scale=\s
		]
		\foreach \x in {0,...,\n} {
			\foreach \y in {0,...,\n} {
				\draw (2*\x,2*\y) rectangle (2*\x+1,2*\y+1);
				\draw[line width=.5mm,\tm] (2*\x,2*\y+1) -- (2*\x+1,2*\y+1);
			};
		};
		\end{scope}
		\end{scope}
	}
	\fill[white,opacity=.75] (-1,42.5) rectangle (16.5,48);
	\node[right] at (.8,45) {#5};
	\end{tikzpicture}
}
\begin{figure}[hbp]
\robinson{nohalt}{nohalt}{nohalt}{nohalt}{a) non-halting}\\[.4cm]
\robinson{nohalt}{halt}{halt}{halt}{b) halting}\\[.4cm]
\tikz{\node at (0,0) {c) non-halting}; \node[white] at (\linewidth-2.33cm, 0) {.}; } \\
\includegraphics[width=8.1cm]{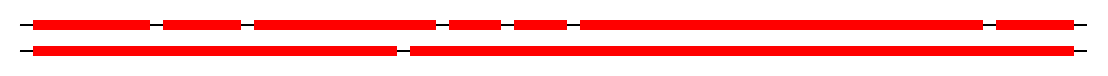}\\[.4cm]
\tikz{\node at (0,0) {d) halting}; \node[white] at (\linewidth-2cm, 0) {.}; }\\
\includegraphics[width=8.1cm]{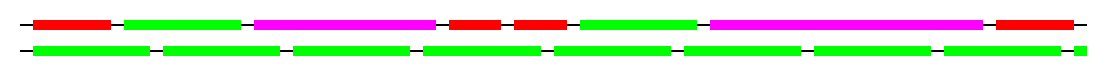}\\[.4cm]
\caption{2D Robinson tiling construction with instances of a Turing machine running on the upper edges of the fractal rectangles.
	Each edge represents the available tape for the Turing machine. In the non-halting case a), there will never be any halting penalty, no matter how much tape there is available.
	In the halting case b), there is a threshold side length after which each rectangle larger than the threshold contributes a penalty (red)---which yields a small but nonzero ground state energy density; the ground state energy diverges.\\
	In the 1D case we show the segments emerging from the Marker Hamiltonian from \cref{sec:marker} (cf.\ \cref{fig:1d}). In the non-halting case c), no segment length is long enough to contain the entire computation; all segments obtain a penalty (red).
	In the halting case d), there is an ideal segment length (green) with \emph{just} enough tape tor the TM to halt; as per \cref{fig:pen-and-bonux}, this segment has the maximum possible bonus. Segments too short (red) contribute a net energy penalty, whereas segments too long (magenta) do contribute a bonus, yet not one as large as the optimal segment length.
}
\label{fig:2d}
\end{figure}
In order to spawn a fixed density of computations in 1D without the aid of a fractal underlying structure, we need to know an optimal segment length to subdivide the spin chain into.
In the halting case, this should be just enough tape for the computation to terminate.
However, if we aim to construct a reduction from the Halting Problem, we cannot know the space required beforehand---which, in particular, could be uncomputably large, or infinite!
One way out is to spawn Turing machines on tapes of all possible lengths, \emph{and} do this with a fixed density.
In 2D this can be achieved using an underlying fractal tiling such as that due to Robinson~\cite{Robinson1971}, see \cref{fig:2d}.

The two-dimensional construction thus crucially depends on one's ability to create structures of all length scales, in order to define ``lines'' of all sizes~\footnote{In fact, sizes $4^n$ for all integer $n$~\cite{Cubitt2015_long}.}, which are then used as a tape for running a Quantum Turing machine: the key property of the fractal which makes the construction work is that every possible tape length indeed appears with a non-zero density in the large system size limit.

As already mentioned, constructing a fractal tiling with a fixed density of structures of all length scales seems impossible in one dimension.
We therefore replace the fractal Robinson tiling with a ``marker'' Hamiltonian, where
the markers bound sections of tape used for the Turing machine (just like the lower boundaries of the squares in \cref{fig:2d}).
We will construct the Hamiltonian in such a way that two consecutive markers bounding a segment will introduce an energy bonus that falls off quickly as the length of the segment increases.
This bonus will be weak enough to permit an executing QTM to ``extend'' the tape as needed, in the sense that the bonus due to the marker boundaries is strictly smaller in magnitude than the potential penalty introduced when the QTM head runs out of tape (see \cref{fig:1d}).

\subsection{The Marker Hamiltonian}
We now construct the Marker Hamiltonian. It will be a local Hamiltonian $\op H$ on a chain of qudits with a special spin state $\ket\bd$, which we call a \emph{boundary}, and which will separate the different tape segments.
For a product state $\ket\psi$, we define a \emph{signature} with respect to these boundaries as the binary string with 1's where boundaries are located, and 0's everywhere else, which we will denote by $\sig\ket\psi$.
The Hamiltonian we construct will leave the signature invariant, i.e.\ $\sig\ket\psi=\sig\op H\ket\psi$ for all $\ket\psi$.
This property allows us to block-diagonalize $\op H$ with respect to states of the same signature.
For a given block signature, say $(1,0,0,0,1,0,0,1)$, the Hamiltonian gives an energy bonus (i.e.\ a negative energy contribution) to each $1$-bounded segment, which is large when the boundary markers are close, and becomes smaller the longer the segment.
This introduces a notion of boundaries that are ``attracted'' to each other, and our goal is to have a falloff as $\sim -1/g(l)$ in the segment's length $l$, where $g$ is a function we can choose.
In brief, ``attraction'', in this context, simply means that the energy bonus given by $\op H$ to pairs of boundary symbols grows the closer they are to each other.

For reasons of clarity, we start by constructing a Hamiltonian where the falloff is a fixed function $g$ that is asymptotically bounded as $\Omega(2^l) \leq g \leq \BigO(4^l)$.
In a second step, we allow the falloff to be tuned, replacing $l$ by an arbitrary exponential in $l$, such that the falloff is \emph{doubly} exponential in the segment length.
\subsection{Construction}
We start with the following lemma.
\begin{lemma}\label{lem:elastic-1}
Let $\Hs\coloneqq (\field C^3)^{\otimes N}$ be a chain of qutrits of length $N$ with local computational basis $\{ \ket\bd, \ket\bl, \ket\hd \}$,
and for a product state $\ket\psi\in\Hs$, $\ket\psi=\ket{\psi_1}\cdots\ket{\psi_N}$, we define a ``boundary signature'' $\sig\ket\psi=(\braket\bd{\psi_1},\ldots,\braket{\bd}{\psi_N})$, extended linearly to $\Hs$.
Define two local Hamiltonian terms
\begin{align*}
	\op h_1 & \coloneqq   \ketbra{\hd}\otimes\big(
    \ket{\bl\bl} - \ket{\hd\bl}\big)\big(\bra{\bl\bl}-\bra{\hd\bl}\big) \\
	\op h_2 & \coloneqq   \big(\ket{\hd\bl} - \ket{\hd\hd}\big)\big(\bra{\hd\bl}-\bra{\hd\hd}\big)\otimes\ketbra{\bd}
\end{align*}
and set $\op h_\mathrm{walk} \coloneqq  \op h_1 + \op h_2$.
Let
\begin{align*}
    \op p &  \coloneqq  2\ketbra{\bd\bd} + 2\ketbra{\bl\hd} + 2\ketbra{\bd\bl}
\end{align*}
Then
\begin{align*}
    \op H_\mathrm{walk} + \op P \coloneqq &\sum_{i=1}^{N-2}\1_{\{1,\ldots,i-1\}}\otimes\op h_\mathrm{walk}\otimes\1_{\{i+3,\ldots,N\}}\\
    +&\sum_{i=1}^{N-1}\1_{\{1,\ldots,i-1\}}\otimes\op p\otimes\1_{\{i+2,\ldots,N\}}
\end{align*}
is a 3-local Hamiltonian which is positive semi-definite, and block-diagonal with respect to the subspaces spanned by states with identical signature $\sig$.
\end{lemma}
\begin{proof}
The first two claims are true by construction.
The Hamiltonian $\op H_\text{walk} + \op P$ is further block-diagonal with respect to $\sig$ because $\sig(\op H_\text{walk} + \op P)\ket\psi=\sig\ket\psi\ \forall \ket\psi\in\Hs$,
as none of the local terms ever affect the subspaces spanned by the boundary symbol $\ket\bd$.
\end{proof}

As a second step, we employ a boundary trick by \textcite{gottesman2009quantum} to ensure that blocks not terminated by a boundary marker have a ground state energy at least $2$ higher than $\bd$-terminated blocks.
It is worth emphasizing that this is not achieved by a term that \emph{only} acts on the boundary, but in a translationally-invariant way, i.e.\ by adding the same one- and two-local terms throughout the chain.
In brief, it exploits the fact that while there are $N$ spins in the chain, there is only $N-1$ edges between them.
We state this rigorously in the following remark.
\begin{remark}[\textcite{gottesman2009quantum}]\label{rem:elastic-1}
  Give an energy \emph{bonus} of strength 4 to $\ket\bd$, and an energy \emph{penalty} of $2$ to $\ket\bd$ appearing next to any symbol (including $\bd$ itself.
  I.e.\ if $\ket\bd$ appears at the end of the chain there will be a net bonus of $2$, otherwise a net penalty of zero).
  Collect these terms in a Hamiltonian $\op P'$.
  Then, apart from positive semi-definiteness,
  \begin{equation}\label{eq:H-rem:elastic-1}
  \op H \coloneqq \op H_\text{walk} + \op P + \op P'
  \end{equation}
  where $\op H_\text{walk}$ and $\op P$ are defined in \cref{lem:elastic-1}, has the same properties claimed in \cref{lem:elastic-1}, but any block not terminated by a boundary will have energy $\ge -2$, while all properly-terminated blocks will have a ground state energy $-4$.
\end{remark}
\begin{proof}
	The first claim is straightforward, as $\op P'$ does not change the interaction structure of $\op H$.
	The last claim follows from the fact that the only way of obtaining a net bonus is to place a boundary symbol at the end of the spin chain, where it picks up a net bonus of $2$.
	The maximum possible bonus of any state is thus $4$, which will be achieved by signatures that are properly bounded on either side.
\end{proof}

From now on, when we talk of ``properly bounded'', we always mean a signature with boundary blocks $\bd$ at each end.
Individual cases where only one side carries a boundary will be mentioned as such explicitly then.

\subsection{Spectral Analysis}
In the following, the ``good'' blocks will therefore be those that have ground space energy $-4$, all of which aFre properly bounded.
\Cref{rem:elastic-1} allows us to analyze the blocks more closely, which we do in the following lemma.
\begin{lemma}\label{lem:elastic-2}
Let $\op H = \op H_\text{walk} + \op P + \op P'$ be as in \cref{eq:H-rem:elastic-1}.
If we write $\op H=\bigoplus_{s\in\{0,1\}^N} \op H_s$ as the block-decomposition of $\op H$, where $s$ denotes an arbitrary length $N$ binary string, then every properly bounded block will either
\begin{enumerate}
	\item have two consecutive boundaries, and thus a ground state energy $\ge -2$, or
	\item have signature of consecutive $1$-bounded segments of $0$s. In this case, $\op H_s$ further block-diagonalizes into $\op H_s=\op G_s\oplus\op R_s$, where $\op R_s$ is within the span of states penalized by $\op P$ in \cref{lem:elastic-1}, and $\op G_s$ in its kernel.
	\item The ground state energy of $\op R_s$ is $\ge-2$.
	\item The ground state energy of $\op G_s$ equals $-4$, and $\op G_s$ will be a sum of terms of the form $\1_{\{1,\ldots,l\}}\otimes\Delta_w\otimes\1_{\{l+w+1,\ldots,N\}}$, where $\Delta_w$ is the Laplacian of a path graph of length $w$ (i.e.\ a graph with vertices $\{1,\ldots,w\}$ and edges $\{(i,i+1):i=1,\ldots,w-1\}$).
	Here $l$ and $w$ depend on the signature $s$---more precisely, for every contiguous section of $0$s in $s$ surrounded by a pair of $1$s, $l$ marks the left $1$ and $w$ is the length of the section of $0$s.
\end{enumerate}
\end{lemma}
\begin{proof}
If there are two neighbouring $1$s in the signature $s$, the penalty term $\ketbra{\bd\bd}$ picks up an energy contribution of 2.
Since $\op H_\mathrm{prop}$ is already positive semi-definite and block-diagonal with respect to signatures, any state $\ket\psi$ with support fully contained in the block corresponding to signature $s$ must thus necessarily satisfy $\bra\psi\op H\ket\psi \ge \bra\psi \op P \ket\psi\ge2$.
The first claim follows.

So let us assume that all $1$s are spaced away from each other with at least one $0$.
Within the 2-dimensional $0$ subspace spanned by the local basis states $\ket\hd$ and $\ket\bl$.
We note that the penalized substring $\ket{\bl\hd}$ is also an invariant, meaning that no transition rule can create or destroy this configuration.
Any state that, when expanded in the computational basis, has at least one expansion term with said substring will thus necessarily have \emph{all} terms with this specific substring.
The same arguments holds for the invariant substring $\ket{\bd\bl}$, and the second claim follows.

Since any eigenstate of $\op R_s$ picks up the full penalty contribution of $2$, the third claim follows.

If neither of the invariant substrings $\ket{\bl\hd}$ and $\ket{\bd\bl}$ occur, we can assume that all $1$-bounded segments of $0$s lie within the span of the states
\begin{align}
    \ket{\hd\bl\bl\cdots\bl\bl}, \ket{\hd\hd\bl\cdots\bl\bl}, \ldots \notag \\ \ldots, \ket{\hd\hd\hd\cdots\hd\bl}, \ket{\hd\hd\hd\cdots\hd\hd}. \label{eq:elastic-states}
\end{align}
Since there is no penalty acting on any of those states, the ground state energy of $\op G_s$ equals $-4$.

Each such segment of contiguous $0$s thus defines a separate path graph, where the vertices are precisely these states, linked by the transition rules given in $\op H_\mathrm{walk}$ in \cref{lem:elastic-1}.
Denote the path graphs corresponding to these segments with $G_1,\ldots,G_n$, where we assume that there are $n$ 1-bounded segments of 0s in signature $s$.
As each segment is independent of the others,
the overall graph spanned by these individual paths is the Cartesian product of the individual paths, i.e.\ $G=G_1\square G_2\square\ldots\square G_n$.
This is precisely a hyperlattice with side lengths uniquely determined by the lengths of the individual segments.

The transition rules in $\op h_\mathrm{walk}$ therefore result in a block $\op G_s=\Delta_G$, i.e.\ the Hamiltonian is precisely the Laplacian of the graph of determined by the transition rules (for an extensive analysis see e.g.~\cite{Bausch2016}).
We further know that the Laplacian of a Cartesian product of graphs decomposes as
\begin{align}
    \Delta(G)&=\Delta(G_1)\otimes\1\otimes\ldots\otimes\1 \notag \\ &+ \1\otimes\Delta(G_2)\otimes\1\otimes\ldots\otimes\1 + \ldots \notag \\ \ldots
    &+ \1\otimes\ldots\otimes\1\otimes\Delta(G_n), \label{eq:cartesian-dec}
\end{align}
and the last claim follows.
\end{proof}
A more direct route to \cref{eq:cartesian-dec} is to note that $\op H_\mathrm{walk}$ is by definition the Laplacian of a graph with vertices given by strings of the alphabet $\{\bd,\bl,\hd\}$, and edges by the transition rules in \cref{lem:elastic-1}.
Those connected graph components that do not carry a penalty due to an invalid configuration (which either holds for \emph{all} vertices, or \emph{none}) are lattices in $n$ dimensions---where $n$ is the number of $1$-bounded segments---and side lengths determined by the segments' lengths.
\Cref{eq:cartesian-dec} is precisely the Laplacian of this grid graph.

For the sake of clarity, we will keep calling the segments of consecutive zeros bounded by $\bd$ on either side ``$1$-bounded segments'', and when talking about the entire string we use the term ``properly bounded''.
We will henceforth re-label the states in \cref{eq:elastic-states} as $\ket 1,\ldots,\ket w$, where $w$ denotes the length of the segment.
Our next step will be to add a 2-local bonus term which gives an energy bonus to the arrow appearing to the left of the boundary, i.e.\ to $\ket{\hd\bd}$.
\begin{lemma}\label{lem:elastic-3}
Define $\op H'\coloneqq \op H+\op P''+\op B$, where
\begin{itemize}
	\item $\op H$ is taken from \cref{eq:H-rem:elastic-1},
	\item $\op P''=1/2\sum_{i=1}^N\ketbra{\bd}_i$ gives a penalty of $1/2$ to any boundary term, and
	\item $\op B=-\sum_{i=1}^{N-1}\ketbra{\hd\bd}_{i,i+1}$ gives a bonus of 1 to states where the arrow has reached the right boundary.
\end{itemize}
Then
\begin{enumerate}
	\item $\op H'$ is still 3-local and block-diagonal in signatures, i.e.\ $\op H'\coloneqq \sum_s\op H'_s$.
	If $s$ is properly bounded and has no double $11$s, the corresponding block decomposes as $\op H'_s=\op G'_s\oplus \op R'_s$ similar to \cref{lem:elastic-2}, but such that the primed versions carry the extra penalties and bonus terms.
	\item For any such $s$, $\op R'_s\ge\op G'_s+2$.
	\item $\op G'_s$ breaks up into sum of terms of the form $\1\otimes\Delta'_w\otimes\1$, where $\Delta'_w$ is a perturbed path graph Laplacian $\Delta'_w\coloneqq \Delta_w - \ketbra{w}$ (where $\ket w$ labels the last of the basis states given in \cref{eq:elastic-states}, as mentioned).
\end{enumerate}
\end{lemma}
\begin{proof}
The first two claims follow immediately from \cref{lem:elastic-2}, since all of the newly-introduced terms leave signatures and penalized substrings invariant, and are at most 2-local.

Since the Cartesian graph product is associative and commutative, it is enough to show the decomposition for the case of two graphs $G_1$ and $G_2$, and a single vertex $v\in G_1$ which we want to give a bonus of $-1$ to. Denote the bonus matrix for $G_1$ with $\op B_1$.
We have that the adjacency matrix $\op A_{G_1\square G_2}=\op A_{G_1}\otimes\1 + \1\otimes \op A_{G_2}$.
Vertex $v$ is thus mapped to a family of product vertices $(v,v')_{v'\in G_2}$, which are precisely the corresponding bonus'ed vertices in $G=G_1\square G_2$ that have to receive a bonus of $-1$.
The bonus term for $G$ is thus $\op B=\op B_1\otimes\1$, and the claim follows.
\end{proof}

We know that any Laplacian eigenvalues $\mu,\nu$ of two graphs $G_1,G_2$ combine to a Laplacian eigenvalue $\mu+\nu$ of $G_1\square G_2$ (see e.g.\ \cite[Ch.~1.4.6]{Brouwer2012}).
It is straightforward to extend this fact to the case of bonus'ed graphs,
which will allow us to analyse the spectrum of each signature block $\op H'_s$.

The reader will have noticed that in contrast to \cref{lem:elastic-2}, \cref{lem:elastic-3} does not make any claims about the ground state energy of the individual blocks.
Na\"ively, one could assume that the ground state energy of each block will diverge to $+\infty$ with the number of boundaries present, as each of them carries a penalty of $+1/2$---but how does this balance with the bonus of $-1$, which we apply to only a \emph{single} basis state in the graph Laplacian's ground space, and not on each vertex?

In order to answer this question, let us step back for a moment and develop a bound for the lowest eigenvalue of a modified path graph Laplacian $\Delta'_w$.
We will do this in a series of technical lemmas.
\begin{lemma}
    $\Delta'_w$ has precisely one negative eigenvalue.
\end{lemma}
\begin{proof}
Assume this is not the case. Then there exist at least two eigenvectors $\ket u,\ket v$ with negative eigenvalues, and any $\ket x\in\spn\{\ket u,\ket v\}$ satisfies $\bra x\Delta'_w \ket x<0$.  Since $\dim\ker\ketbra w=w-1$,
there exists a nonzero $\ket x\in\spn\{\ket u,\ket v\}$ such that $\ketbra w\ket x=0$. Therefore $0>\bra x\Delta'_w\ket x=\bra x\Delta_w\ket x$, contradiction, since $\Delta_w$ is positive semi-definite.
\end{proof}

As a next step, we will lower-bound the minimum eigenvalue of $\Delta'_w$.
\begin{lemma}
    The minimum eigenvalue of $\Delta'_w$ satisfies $\lambda \ge -1/2-2^{-w}$.
\end{lemma}
\begin{proof}
    We first observe that $\Delta'_w$ is tridiagonal, e.g.
    $$
        \Delta'_5=\begin{pmatrix}
        	1  & -1 & 0  & 0  & 0  \\
        	-1 & 2  & -1 & 0  & 0  \\
        	0  & -1 & 2  & -1 & 0  \\
        	0  & 0  & -1 & 2  & -1 \\
        	0  & 0  & 0  & -1 & 0
        \end{pmatrix}.
    $$
    We can thus expand the determinant $p_w(\lambda)\coloneqq \det(\Delta'_w-\lambda\1)$ using the continuant recurrence relation (see \cite[Ch.~III]{Muir1882})
    \begin{align*}
    	f_0          & \coloneqq 1                             \\
    	f_1          & \coloneqq \lambda-1                     \\
    	f_i          & \coloneqq (\lambda-2) f_{i-1} - f_{i-2} \\
    	p_w(\lambda) & \coloneqq \lambda f_{w-1} - f_{w-2}
    \end{align*}
    As can be easily verified, a solution to this relation is given by the expression
    \begin{equation}\label{eq:pn}
    p_w(\lambda)=-\frac{2^{-w-1}}{\sqrt{\lambda -4}}\left(
        3\sqrt\lambda z_w^-(\lambda)
        +\sqrt{\lambda-4} z_w^+(\lambda)
    \right)
    \end{equation}
    where $z_w^+(\lambda)\coloneqq x_w(\lambda) + y_w(\lambda)$, $z_w^-(\lambda)\coloneqq x_w(\lambda) - y_w(\lambda)$, and
    \begin{align*}
    x_w(\lambda) &= \left(\lambda -\sqrt{\lambda -4} \sqrt\lambda-2\right)^w \\
    y_w(\lambda) &= \left(\lambda +\sqrt{\lambda -4} \sqrt\lambda-2\right)^w.
    \end{align*}
    There is of course no hope to resolve $p_w(\lambda)=0$ for $\lambda$ directly, so we go a different route.
    First note that $p_w(\lambda)$ is necessarily analytic, since it is the characteristic polynomial of $\Delta'_w$.
    We can calculate
    $
    p_w(-1/2)=(-1)^{1 + w} 2^{-w},
    $
    and thus know that $\sign p_w(-1/2)=1$ for $w$ odd, and $-1$ for $w$ even.
    If we can show that $p_w(-1/2-1/2^w)$ has the opposite sign, then by the intermediate value theorem we know there has to exist a root on the interval $[-1/2-1/2^w,-1/2]$, and the claim follows.

    First substitute $p_w(-1/2-1/2^w)=:A_w/B_w$, where
    \begin{align*}
    	B_w     & = 2^{w+1} \sqrt{2^{-w}+\frac{9}{2}},                                                       \\
    	A_w     & = -a_{1,w}(x'_w - y'_w) - a_{2,w}(x'_w + y'_w),                                              \\
    	a_{1,w} & = 3 \sqrt{2^{-w}+\frac{1}{2}},                                                            \\
    	a_{2,w} & = \sqrt{2^{-w}+\frac{9}{2}},                                                              \\
    	x'_w    & = \left(\sqrt{2^{-w}+\frac{9}{2}} \sqrt{2^{-w}+\frac{1}{2}}-2^{-w}-\frac{5}{2}\right)^w, \\
    	y'_w    & = \left(-\sqrt{2^{-w}+\frac{9}{2}} \sqrt{2^{-w}+\frac{1}{2}}-2^{-w}-\frac{5}{2}\right)^w.
    \end{align*}
    Then $B_w$, $a_{1,w}$ and $a_{2,w}$ are real positive for all $w$. We distinguish two cases.
    \paragraph{w even.} If $w$ is even, we need to show $p_w(-1/2-1/2^w)\ge0$, which is equivalent to
    \begin{align*}
    	                         & 0 \le \frac{A_w}{B_w}  \\
    	\Longleftrightarrow\quad & 0 \le A_w=- a_{1,w}(x'_w - y'_w) - a_{2,w}(x'_w + y'_w)                          \\
    	\Longleftrightarrow\quad & 0 \ge a (x'_w-y'_w) + (x'_w+y'_w) \\
    	\Longleftrightarrow\quad & \frac{a-1}{a+1}y'_w\ge x'_w,
    \end{align*}
    where we defined $a\coloneqq a_{1,w}/a_{2,w}\in[1,2]$.
    Now, for $w$ even, $y'_w\ge x'_w$, so it suffices to show
    $$
        \frac{a-1}{a+1}\left(\frac52+\frac32\right)^w\ge\left(\frac52-\frac32 \right)^w
        \quad\Longleftrightarrow\quad
        \frac{a-1}{a+1}\ge \frac{1}{4^w},
    $$
    which is true for all $w\ge2$.

    \paragraph{w odd.} Unlike the even case now we have $y'_w\le x'_w$, and it suffices to show
    $$
        \frac{a-1}{a+1}\left(\frac52\right)^w\le\left(\frac52 \right)^w
                \quad\Longleftrightarrow\quad
                \frac{a-1}{a+1}\le 1,
    $$
    which also holds true for all $w\ge0$.
    This finishes the proof.
\end{proof}

And finally, using a similar approach, we will obtain an upper bound for the minimum eigenvalue of $\Delta'_w$.
\begin{lemma}
    The minimum eigenvalue of $\Delta'_w$ satisfies $\lambda \le -1/2-4^{-w}$.
\end{lemma}
\begin{proof}
    The idea is to extend the area around $-1/2$ for which $p_w$ is positive for $w$ odd, and negative for $w$ even, respectively.
    We start with $p_w$ from \cref{eq:pn}, and substitute $p_w(-1/2-1/4^w)=:A_w/B_w$, where---almost as above, but replacing $2^{-w}$ by $4^{-w}$---we have
    \begin{align*}
    	B_w     & = 2^{w+1} \sqrt{4^{-w}+\frac{9}{2}},\quad\quad\text{\footnotesize(the $2^{w+1}$ is \emph{not} a typo)} \\
    	A_w     & = -a_{1,w}(x'_w - y'_w) - a_{2,w}(x'_w + y'_w),                                                          \\
    	a_{1,w} & = 3 \sqrt{4^{-w}+\frac{1}{2}},                                                                        \\
    	a_{2,w} & = \sqrt{4^{-w}+\frac{9}{2}},                                                                          \\
    	x'_w    & = \left(\sqrt{4^{-w}+\frac{9}{2}} \sqrt{4^{-w}+\frac{1}{2}}-4^{-w}-\frac{5}{2}\right)^w,             \\
    	y'_w    & = \left(-\sqrt{4^{-w}+\frac{9}{2}} \sqrt{4^{-w}+\frac{1}{2}}-4^{-w}-\frac{5}{2}\right)^w.
    \end{align*}
    Then $B_w$, $a_{1,w}$ and $a_{2,w}$ are real positive for all $w$. We distinguish even and odd cases.
   \paragraph{w even.}
   If $w$ is even, we want to show that $p_w(-1/2-1/4^w)\le0$, which is equivalent to
   \begin{align*}
       & 0 \ge \frac{A_w}{B_w} \\
       \Longleftrightarrow\quad & 0 \ge - A_w = - a_{1,w}(x'_w-y'_w) - a_{2,w}(x_w' + y_w') \\
       \Longleftrightarrow\quad & 0 \le a (x'_w-y'_w) + (x'_w + y'_w) \\
       \Longleftrightarrow\quad & \frac{a-1}{a+1}y'_w \le x'_w.
   \end{align*}
   Where again we defined $a\coloneqq a_{1,w} / a_{2,w}\in[1,2]$.
   For $w$ even, $y'_w\ge x'_w$ as before, so we cannot continue as before. Note that, for all $w\ge0$,
   $$
       \frac 12 \left(4^{-w}+\frac{5}{2}\right) \leq \sqrt{4^{-w}+\frac{1}{2}} \sqrt{4^{-w}+\frac{9}{2}}\leq 4^{-w}+\frac{5}{2}.
   $$
   and therefore
   \begin{align*}
       y_w'&\le 2^w\left(4^{-w}+\frac52\right)^w \\
       x_w'&\ge 2^{-w}\left(4^{-w}+\frac52\right)^w.
   \end{align*}
   It thus suffices to show
   \begin{align*}
       \frac{a-1}{a+1}\times2^w\left(4^{-w}+\frac52\right)^w &\le 2^{-w}\left(4^{-w}+\frac52\right)^w\\
       \Longleftrightarrow\quad
       \frac{a-1}{a+1} &\le \frac{1}{4^w}.
   \end{align*}
   It is straightforward to verify that this inequality holds for all $w$.

   \paragraph{w odd.}
   For odd $w$, $y'_w\le x'_w$.  Analogously to before one can show
   \begin{align*}
       y_w'&\ge -2^w\left(4^{-w}+\frac52\right)^w \\
       x_w'&\le -2^{-w}\left(4^{-w}+\frac52\right)^w.
   \end{align*}
   Canceling the minus signs flips the inequality sign, and reduces the odd case to what we have shown for $w$ even. The claim follows.
\end{proof}

We summarize these findings in the following corollary.
\begin{corollary}\label{cor:elastic}
The spectrum of $\Delta'_w$ is contained in $(-1/2 - 1/2^w, -1/2 - 1/4^w) \cup [0,\infty)$.
\end{corollary}

Let us now analyse what this means for the spectrum of $\op H'$.
We are only interested in those blocks $\op G'_s$ which correspond to modified grid Laplacians---all other cases are bounded away by a constant in \cref{lem:elastic-3}.
In brief, the answer will be that the negative energy shift of $-1/2$ in \cref{cor:elastic} will be precisely offset by the shift of $1/2$ for any occurrence of the boundary state $\ket\bd$.

Combining \cref{lem:elastic-3} with \cref{cor:elastic}, we obtain the following theorem.
\begin{theorem}\label{th:elastic}
Let $\op H'$ be as in \cref{lem:elastic-3}.
If $\op H'=\bigoplus_{s\in\{0,1\}^N} \op H'_s$ is the decomposition of $\op H'$ into signature blocks, the following holds.
\begin{enumerate}
\item If $s$ is not properly bounded, i.e.\ where one or both ends have no boundary marker, adding a $\bd$ there (either by adding one explicitly, or moving one from a site one away from the end) yields a signature $s'$ such that $\op H'_s\ge\op H'_{s'} + 1$.
\item If $s$ has two consecutive boundaries, one can always delete one of them and obtain a signature $s'$ such that $\op H'_s\ge \op H'_{s'} + 1$.
\item If $s$ is bounded and without consecutive boundaries, $\op H'_s=\op G'_s+\op R'_s$ as in \cref{lem:elastic-3}.
  In that case, the minimum eigenvalue $\lambda$ of $\op G'_s$ satisfies $-\sum_i 1/2^{w_i} \le \lambda+7/2 \le -\sum_i 1/4^{w_i}$, where $w_i$ is the length of the $i$\textsuperscript{th} contiguous $0$-segments in the signature $s$.
  In that case, furthermore, $\op G'_s$ has a spectral gap of size $\ge1/2$.
\end{enumerate}
\end{theorem}
\begin{proof}
  Claim 1 can be shown by explicitly considering an arbitrary signature, but with one missing boundary.
  We will only discuss the left boundary.
  The right then immediately follows from the fact that one could at most \emph{gain} an extra bonus there from $\op B$ in \cref{lem:elastic-3}.

  First consider the case that the left boundary looks like $s=01\cdots$.
  By moving the boundary from the site to its right, we either break up a double boundary (in case $s=011\cdots$), or enlarge a segment (in case $s=010\cdots01\cdots$).
  In the first case, we obtain i) a net bonus of $2$ by \cref{rem:elastic-1}, ii) a net bonus of $2$ from breaking up a double boundary from \cref{lem:elastic-1}, iii) a bonus $>0$ from creating a $1$-bounded segment.
  In the second case, we also obtain i), but decrease the bonus from the segment to its right.
  This can at most be a penalty of $1/2$, though, and the claim follows.

  Claim 2 can be broken up in cases as well.
  Assume the double boundary is either on the left, or right (e.g.\ $s=110\cdots$).
  By deleting the second site boundary, one obtains a net bonus of at least 1.
  The same holds true for a site in the middle, as can be easily seen.

  Claim 3 follows from \cref{cor:elastic,lem:elastic-3}.
  Every $1$-bounded segment is terminated by a boundary, whose penalty of $1/2$ from \cref{lem:elastic-3} precisely offsets the $-1/2$ shift of the ground state of $\Delta'_w$.
  The leftover overall energy shift of $-7/2$ stems from the original $-4$ ground state from \cref{rem:elastic-1}, and the single penalty of the left boundary of magnitude $1/2$.
  The gap claim follows from \cref{lem:elastic-3} (i.e.
  that $\op R'_s\ge\op G'_s+2$) and the spectral gap of $\Delta'_w$.
\end{proof}

\subsection{A Marker Hamiltonian with a Quick Falloff}
The transition rules in \cref{lem:elastic-1} are those of a unary counter, as depicted in \cref{eq:elastic-states}.
It is clear that if we allow for an increase in the local dimension we can use more complicated transition rules---and assume that they are 2-local---to model the evolution of a more sophisticated calculation (e.g.\ the binary counter construction of~\cite{Cubitt2015_long}, or the Quantum Thue System constructions of~\cite{Bausch2016}).
Instead of the linear exponential dependence on the segment length $w$ in \cref{th:elastic}, we then have the following theorem.
\begin{theorem}[Marker Hamiltonian]\label{th:elastic!}
Take a Hamiltonian $\op H'$ as in \cref{th:elastic}, but with $2$-local transition rules describing a path graph evolution of length $f(w)$ on a segment of length $w$.
Furthermore, we add an energy shift of $7/2$ by adding a term $$7/2\sum_{i=1}^N\1_{\{i\}} - 7/2\sum_{i=1}^{N-1}\1_{\{i,i+1\}}.$$
Denote this Hamiltonian with $\Hel$.
Then $\Hel=\bigoplus_s\Hel_s$ as before.
We have $\Hel_0\ge0$, and either $\Hel_s\ge1/2$, or its minimum eigenvalue satisfies
$$
-\sum_i 1/2^{f(w_i)} \le \lambda \le -\sum_i 1/4^{f(w_i)},
$$
where $w_i$ is the $i$\textsuperscript{th} segment length.
\end{theorem}
\begin{proof}
Precisely the same argument as in the proof of \cref{th:elastic}, taking into account an energy shift of $+7/2$ due to the mismatch in the number of one-local and two-local couplings available in a system with open boundary conditions, see \cref{rem:elastic-1}.
\end{proof}

We conclude with the following two remarks.
\begin{remark}\label{rem:elastic!}
On a spin chain with nearest neighbour interactions and local dimension $d$ (including the boundary symbol~$\bd$), one can obtain a path graph evolution length $f(w)=(d-c_1)w$, or alternatively $f(w)=(d-c_2)^w$, where $c_1$ and $c_2$ are constant.
Each signature block $\Hel_s$ of the corresponding Hamiltonian thus has a unique lowest-energy eigenvalue
$$
-\sum_i 1/2^{(d-1)w_i} \le \lambda \le -\sum_i 1/4^{(d-1)w_i}
$$
or
$$
\ \ -\sum_i 1/2^{(d-5)^{w_i}} \le \lambda \le -\sum_i 1/4^{(d-5)^{w_i}},
$$
respectively, with a spectral gap $\ge1/2$,
where $w_i$ is the $i$\textsuperscript{th} segment length.
\end{remark}
\begin{proof}
A unary counter does not require any special head symbols (see e.g.~\cite{Kitaev2002})
It is further known that one can construct an arbitrary base counter with four additional symbols (see e.g.~\cite{gottesman2009quantum}).
Breaking either of the constructions down to 2-local at most adds a constant overhead.
The rest follows from \cref{th:elastic!}.
\end{proof}

\begin{remark}\label{rem:min-and-even}
    Increasing the local dimension by a constant factor $d_1$ allows us to add two-local penalty terms to $\Hel$, which enforce that the only blocks $\Hel_s$ with negative ground state energies as in \cref{th:elastic!} have minimum segment length $w_i\ge d_1$.
    Similarly, increasing the local dimension by another constant factor $d_2$ allows us to assume segment lengths $w_i \equiv 0\pmod{d_2}$.
\end{remark}
\begin{proof}
	In the first case, we impose that each boundary term is followed by a sequence of states $\ket 0,\ket 1,\ldots,\ket{d_1}$, the latter of which we allow to be followed by $\ket{d_1}$ only.
	Now penalize a boundary term to the right of anything but $\ket{d_1}$.

	The second proof is similar, where instead of counting once we count modulo $d_2$, and penalize the boundary state to appear to the right of anything but $\ket{d_2}$.
\end{proof}

\section{\label{sec:qpe}Augmented Phase Estimation QTM}
\subsection{Phase estimation}
Just as in the two-dimensional case, we will use a phase estimation QTM to extract the input to a universal TM from the phase of a specific gate.
This is the only ingredient we will require from \textcite{Cubitt2015_long}; yet in addition to the original construction, we will need to be able to detect and penalize the case where the phase estimation does not terminate with the full binary expansion.
This can be done with a slight modification to the original procedure from~\cite[sec.~6]{Cubitt2015_long}.

For completeness and for self-consistency we state the relevant results from \cite[sec.~6]{Cubitt2015_long} in the following.

\begin{theorem}[Phase-estimation QTM (\textcite{Cubitt2015_long})]\label{th:phase-QTM}
There exists a family of QTMs $P_\ivar$ indexed by $\ivar\in\field N$, all with identical internal states and symbols but differing transition rules, with the property that on input $N\ge|\ivar|$ written in unary, $P_\ivar$ halts deterministically after $\BigO(\poly(N)2^N)$ steps, uses $N+3$ tape, and outputs the binary expansion of $\ivar$ padded to $N$ digits with leading zeros.
\end{theorem}
As the authors state, it is crucial that $N$ does not determine the binary expansion that is written to the tape, only the number of digits in the output.
The authors construct this family of QTMs explicitly, in three parts:
\begin{enumerate}
\item Apply the controlled $\op U^k$-gates, where $\op U$ is the phase gate encoding $\ivar$ (see \cref{fig:phase-estimation}).
\item Detect the least significant bit.
\item Perform an inverse quantum Fourier transform (see \cref{fig:inverse-qft}).
\end{enumerate}

The problem with using this series of steps unchanged is linked to the fact that we cannot apply the standard inverse quantum Fourier transform, for two reasons.
First, we need the result of the QFT to be exact---so using approximate QFT is not an option.
This in turn would imply we need an infinite local dimension, as we need a potentially infinite set of controlled phase gates.
In the 2D construction, it suffices for the authors to provide a phase gate with minimum rotation $\alpha=2^{-|\ivar|}$, since the case of too-short-segments can be independently detected there (see \cite[sec.\ 5.3]{Cubitt2015_long} for an extensive discussion).

However, in 1D, we cannot \`a priori know whether there is enough tape space for the full expansion, so finding the least significant bit is not always possible.
A simple solution is as follows.
By \cref{rem:min-and-even}, we can always assume that the tape has length at least $10$, and $\equiv 0\pmod 2$.
We can then encode the input $\ivar$ as follows:
\begin{equation}\label{eq:phi-encoding}
\ivar= \ivar_1 \ivar_2\cdots \ivar_{|\ivar|} \xmapsto{\text{enc}} \phi(\ivar) = \phi\coloneqq  \ivar_11 \ivar_21\cdots1 \ivar_{|\ivar|}0,
\end{equation}
i.e.\ we interleave the bits of $\ivar$ with $1$s.
In this way, by always reading pairs of bits, we know that once the second bit is $0$, all digits of $\phi$ have been extracted.
In the following, we will assume that all inputs $\phi$ are always in the form \cref{eq:phi-encoding}.

The quantum phase estimation procedure can then be modified as follows.
\newcommand{\bolditem}{\normalfont\bfseries\origitem}
\begin{enumerate}
\item Apply the controlled $\op U^k$-gates, where $\op U$ is the phase gate encoding $\ivar$ (see \cref{fig:phase-estimation}).
\textbf{\item Move the head to the least significant bit on the tape, and transition to a unique head symbol there.}
\item Detect the least significant bit.
\item Perform the inverse quantum Fourier transform (see \cref{fig:inverse-qft}).
\end{enumerate}

Steps 1, 3 and 4 are unchanged.
In the next two sections we will rigorously show how this modification suffices to signal expansion success, and penalize all segments with insufficient space for the full expansion.

\subsection{Expansion-Success-Signalling Quantum Phase Estimation}
As a first step, we consider the requirement that the input $N$ written in unary on the tape is longer than $|\phi|+3$.
The tape is the space between two boundary symbols on a segment.
As such, the segment length determines the maximum unary number $N$ that we can write on the tape initially.
Since we cannot \`a priori lower-bound the segment length to guarantee that $N\ge|\phi|+3$, we have to consider the case $N<|\phi|+3$.

We will analyze the behaviour of this by going through the explicit construction of~\cite[sec.~6]{Cubitt2015_long} step by step, and analyse how a too-small $N$ affects the program flow.
The phase estimation QTM is defined on the tape, but such that the tape has multiple tracks: a quantum track, where the quantum operations are performed, as well as classical tracks which are used for the control logic of the QTM---we refer the reader to~\cite[sec.\ 6.1.1\&6.2]{Cubitt2015_long} for details.
The QTM follows five steps.
\paragraph{Preparation Stage.}
The first cell of the quantum track is the ancilla qubit for the phase estimation, and the following $N$ cells are the output qubits for the phase estimation.
\begin{enumerate}
\item Copy the quantum track's unary $1\cdots1$ to a separate input track, in binary.
This TM can work within a length $N+1$ tape (\cite[lem.~30]{Cubitt2015_long}), so there is no issue with this step.
We can thus assume that the separate input track contains the number $N$ written in binary, and padded with $0$s.
\item The $N+1$ qubits in the quantum track are then initialized to $\ket1(\ket+)^{\otimes N}$.
Again, there is no issue.
\end{enumerate}

\paragraph{Control-Phase Stage.}
\newcommand{\ggate}[1]{\gate{\makebox(28,14){$#1$}}}
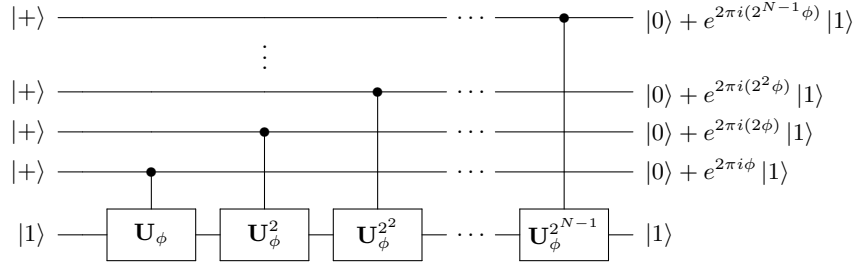
\begin{figure*}
  \centering
  \mbox{
    \Qcircuit @C=1em @R=1.4em {
      \lstick{\ket{+}} & \qw & \qw              & \qw                & \qw                    & \push{\makebox[2em]{$\,\dots$}} \qw & \ctrl{5}                  &  \rstick{\ket{0} + e^{2\pi i(2^{N-1}\phi)}\ket{1}} \qw \\
      \raisebox{4pt}{\hspace{5.5cm}\vdots} &  \\
      \lstick{\ket{+}} & \qw & \qw              & \qw                & \ctrl{3}               & \push{\makebox[2em]{$\,\dots$}} \qw & \qw                       &  \rstick{\ket{0} + e^{2\pi i(2^2\phi)}\ket{1}} \qw \\
      \lstick{\ket{+}} & \qw & \qw              & \ctrl{2}           & \qw                    & \push{\makebox[2em]{$\,\dots$}} \qw & \qw                       &  \rstick{\ket{0} + e^{2\pi i(2\phi)}\ket{1}} \qw \\
      \lstick{\ket{+}} & \qw & \ctrl{1}         & \qw                & \qw                    & \push{\makebox[2em]{$\,\dots$}} \qw & \qw                       & \rstick{\ket{0} + e^{2\pi i\phi}\ket{1}} \qw \\
      \lstick{\ket{1}} & \qw & \ggate{\op U_\phi} & \ggate{\op U_\phi^2} & \ggate{\op U_\phi^{2^2}} & \push{\makebox[2em]{$\,\dots$}} \qw & \ggate{\op U_\phi^{2^{N-1}}}& \rstick{\ket{1}} \qw
    }
    \hspace{4em}
  }
  \caption{Quantum phase estimation, controlled phase gate stage. Figure taken from~\cite{Cubitt2015_long}, but with Hadamards already applied.}
  \label{fig:phase-estimation}
\end{figure*}

\newcommand{\sggate}[1]{\gate{\makebox(19,14){$#1$}}}
\begin{figure*}
  \centering
  \hspace{1cm}\mbox{
    \Qcircuit @C=.5em @R=1em {
      \lstick{\ket{0} + e^{2\pi i(2^{N-1}\phi)}\ket{1}}
        & \qw & \qw & \qw & \qw & \qw & \qw & \qw & \qw & \qw
        & \sggate{\op U_\alpha^1}
        & \sggate{\op U_\alpha^2}
        & \push{\makebox[2em]{$\,\dots$}} \qw
        & \sggate{\op U_\alpha^{2^{N}}}
        & \sggate{\makebox[2em][c]{$\op H$}}
        & \rstick{\ket{j_1}} \qw \\
      \lstick{\ket{0} + e^{2\pi i(2^{N-2}\phi)}\ket{1}}
        & \qw & \qw & \qw & \qw & \qw
        & \sggate{\op U_\alpha^2}
        & \sggate{\op U_\alpha^{2^N}}
        & \push{\makebox[2em]{$\,\dots$}} \qw
        & \sggate{\makebox[2em][c]{$\op H$}}
        & \qw & \qw & \qw & \ctrl{-1} & \qw
        & \rstick{\ket{j_2}} \qw \\
      \raisebox{5pt}{\vdots}\\
      \lstick{\ket{0} + e^{2\pi i(2\phi)}\ket{1}}
        & \qw & \qw
        & \sggate{\op U_\alpha^{2^{N}}}
        & \sggate{\makebox[2em][c]{$\op H$}}
        & \push{\makebox[2em]{$\,\dots$}} \qw
        & \qw & \ctrl{-2} & \qw & \qw & \qw & \ctrl{-3} & \qw & \qw & \qw
        & \rstick{\ket{j_{N-1}}} \qw \\
      \lstick{\ket{0} + e^{2\pi i\phi}\ket{1}}
        & \qw
        & \sggate{\makebox[2em][c]{$\op H$}}
        & \ctrl{-1} & \qw
        & \push{\makebox[2em]{$\,\dots$}} \qw
        & \ctrl{-3} & \qw & \qw & \qw & \ctrl{-4} & \qw & \qw & \qw & \qw
        & \rstick{\ket{j_{N}}} \qw
    }
  }
  \caption{Quantum phase estimation, inverse Fourier transform stage.
    Here, $\alpha=2^{-|\phi|}$, as in~\cite{Cubitt2015_long}.
    This allows us to only have a finite set of gates in the Hamiltonian, instead of requiring an arbitrarily small gate with angle $2^{-N}$.
    Note that this crucially depends on the ability to detect the least significant bit from the control-phase stage.}
  \label{fig:inverse-qft}
\end{figure*}
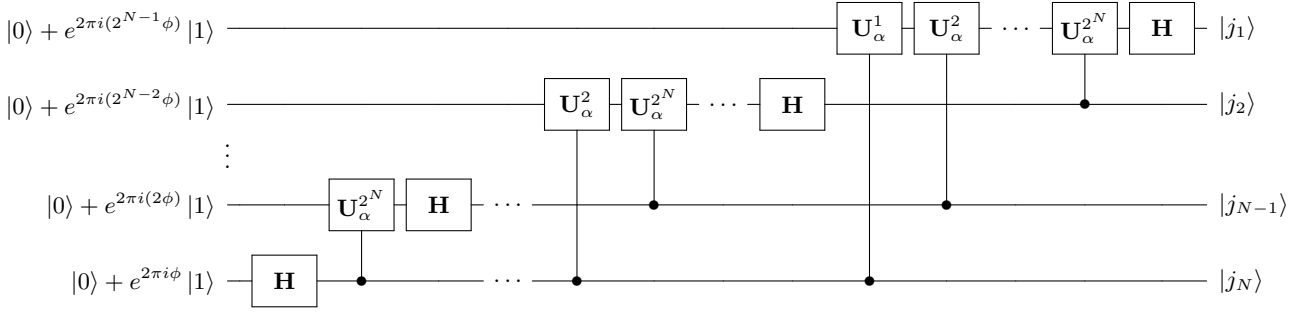

This stage applies the first part of the phase estimation algorithm shown in \cref{fig:phase-estimation}.
It is crucial to note here that just because the input size $N$ is not long enough to do the full phase estimation, the algorithm which is applied is still run as intended for $N$ steps.

If $\phi$ has binary expansion $\phi=0.\phi_1\cdots\phi_{|\ivar|}$, then the output on the first $N$ qubits is
\begin{align}\label{eq:controlled-u-state}
\ket\Phi=\frac{1}{2^{N/2}} \prod_{j=1}^N\left(\ket 0+\ee^{2\pi\ii 2^{N-j}\phi}\ket 1\right).
\end{align}

\paragraph{Signalling Expansion Success}
Since we only want to consider the full binary expansion of $\phi$ as a good input for the dovetailed universal TM, we need to have a way of signaling whether the full expansion has been delivered, or only a truncated version.
We know that in \cref{eq:controlled-u-state}, the first qubit will be in state $\ket+$ if and only if the expansion happened in full.
This is captured in the following lemma.
\begin{lemma}\label{cor:phase-signal}
	If we assume the phase $\phi$ in \cref{th:phase-QTM} to be interleaved with $1$s and terminating with a $0$ as in \cref{eq:phi-encoding}, and if $N$---the number of expansion bits---was even,
	the state post the controlled-$\op U^k$ stage, \cref{eq:controlled-u-state}, has the following properties:
	\begin{enumerate}
		\item If $N\ge|\phi|+1$, then $|\bra{-}(\ket 0+\ee^{2\pi\ii 2^{N-1}\phi}\ket 1)|^2=0$.
		\item Otherwise---if the phase estimation truncated $\phi$---then $|\bra{-}(\ket 0+\ee^{2\pi\ii 2^{N-1}\phi}\ket 1)|^2=\Omega(2^{-|\phi|})$.
	\end{enumerate}
\end{lemma}
\begin{proof}
The first claim follows since the least significant non-zero digit of $\phi$ is 1 by assumption, so $2\pi 2^{N-1}\phi = 0 \pmod{2\pi}$.

For the second claim there are two extreme cases of $\phi$ to analyze; all others can easily be seen to be bounded by those.
The first case is if there is only one more bit of $1$ past where the expansion happened, i.e.\ a single $1$ that is cut off: $2\pi 2^{N-1}\phi = 0.1\phi_{|\ivar|}0\cdots \pmod{2\pi}$, and $\phi_{|\ivar|}=0$.
Then $(\ket 0+\ee^{2\pi\ii 2^{N-1}\phi}\ket 1)/\sqrt2=\ket-$.
The other case is $2\pi 2^{N-1}\phi = 0.1\cdots10\cdots \pmod{2\pi}$, with $\le |\phi|$ 1s.
Then
\begin{widetext}
\[
\left| \bra-(\ket 0+\ee^{2\pi\ii 2^{N-1}\phi}\ket 1)/\sqrt2 \right|^2 =\frac12 \left| 1-\ee^{2\pi\ii 2^{N-1}\phi} \right|^2 
= 1-\cos(2\pi 0.1\cdots10) 
\ge \left(2\pi\left(\sum_{i=1}^{2|\ivar|}\frac{1}{2^i}-1\right)\right)^2 
=4\pi^2\times 2^{-|\phi|}.
\qedhere
\]
\end{widetext}
\end{proof}

In order to temporarily transition to a specific head state $q_?$ over the leftmost qubit which we just showed to have large overlap with $\ket-$ in case of a truncated output, we dovetail the controlled phase stage with the following trivial machine.
The head state $q_?$ together with the underlying qubit will later allow us to discriminate between the two cases in \cref{cor:phase-signal}.
\begin{lemma}\label{lem:success-signalling-TM}
We can dovetail the Controlled Phase QTM with a QTM $M_s$ with the following properties.
\begin{enumerate}
\item The head sweeps all the way to the end of the tape.
\item The head moves one step to the left.
\item The head changes to a special internal state $q_?$ and moves left.
\item The head changes out of $q_?$ and moves right.
\item The head moves all the way back to the left.
\end{enumerate}
\end{lemma}
\begin{proof}
	Observe that after the reset stage in \cite[Sec.~6.7]{Cubitt2015_long}, the input track is in its original configuration, containing $N$ 1s and right-padded with zeros.
	We give the following partial transition table for the Turing machine.
	\begin{center}
	\begin{tabular}{r|ccc}
		         &             \#             &              0               &              1               \\ \hline
		   $q_0$ &                            &                              &  $\ket{\#}\ket{q_1}\ket{R}$  \\
		   $q_1$ &                            &  $\ket{0}\ket{q_?}\ket{L}$   &  $\ket{1}\ket{q_1}\ket{R}$   \\
		   $q_?$ &                            &                              &  $\ket{1}\ket{q_2}\ket{R}$   \\
		$q_{2}$ &                            & $\ket{0}\ket{q_{3}}\ket{L}$ &                              \\
		$q_{3}$ & $\ket{1}\ket{q_f}\ket{N}$  &                              & $\ket{1}\ket{q_{3}}\ket{L}$ \\
		   $q_f$ & $\ket{\#}\ket{q_0}\ket{N}$ &  $\ket{0}\ket{q_0}\ket{N}$   &  $\ket{1}\ket{q_0}\ket{N}$
	\end{tabular}
	\end{center}
	It is easy to check that the rules define a well-formed (orthogonal transition functions where each non-zero transition probability is 1, see \cite[Thm.~19]{Cubitt2015_long}), unidirectional (each state can only be entered from one side, see \cite[Def.~17]{Cubitt2015_long}), proper and normal form (forward transitions from the final state go to the initial state, not moving the head, and not altering the tape, see \cite[Def.~15]{Cubitt2015_long}) QTM.
\end{proof}

\paragraph{Inverse Fourier Transform Stage.}
The inverse Fourier transform is applied to the output of the phase estimation.
It is crucial to observe again that the control flow for the application of the Fourier transform TM does not change behaviour simply because the tape is too short to contain all $|\phi|$ digits of $\phi$.

The trouble is that since we cannot necessarily locate the least significant bit if the expansion was truncated, we possibly apply the ``wrong'' inverse QFT.
Thus, from hereon, we cannot guarantee that the output is related to the input in any way to keep the dovetailed UTM halting, if it were to halt on the fully-expanded $\phi$, or likewise non-halting.
As we have mentioned before, we note that we do not need to care about this problem: we already have an independent state we can penalize ($q_?$ over $\ket-$) in case the QPE truncated the expansion.

\subsection{On Proper QTM Behaviour}
As in the two-dimensional construction, we have to ensure that one can write a valid history state Hamiltonian from the defined quantum Turing machine.
One requirement is that when the QTM is specified by a partial isometry for the transition rules, they can be uniquely completed to a unitary transition function.
In \textcite{Cubitt2015_long}'s case, the authors ensured this by requiring that the QTM was \emph{proper}, as defined in \cite[Def.~20]{Cubitt2015_long}---meaning that the QTM head moves deterministically on a subset of good inputs.
This not only means that there should never be an explicit transition for a head state into a superposition, but also that any intermediate superposition on the quantum tape does not result in the head splitting up into distinct states.
For TM tapes that were too short, the authors could not guarantee this property (just as we cannot here).
This was not an issue in the 2D construction, since the energy contribution from these cases could be obtained by exact diagonalization (the binary length of $\phi$ is known, hence also an upper bound on the too-short-segment length) and subtracted from the final Hamiltonian.

The reason for proper behaviour in the good case---i.e.\ long enough tape---is more subtle.
Assume for now we have a non-halting instance $\phi$.
If the QTM head were to move in some superposition, it could be that on some long but finite track, one head path reaches the boundary.
Since there is no more tape, the clock moves this head to an idling tape.
This head path is thus not able to interfere back with the other head paths.
The other head paths could now think that one has a halting instance, skewing the result.
It is therefore crucial that the QTM we design behaves properly for long enough tapes.
\begin{remark}\label{rem:proper-qtm}
	On a segment which is long enough the QTM plus dovetailed sweeper in \cref{lem:success-signalling-TM} we use is proper, in the sense of \cite[Def.~20]{Cubitt2015_long}.
\end{remark}
\begin{proof}
	The phase estimation terminates with success probability of $1$ if the tape is long enough, and we refer the reader to~\cite[sec.~6]{Cubitt2015_long} for a discussion of the proper QTMs they use, and whose existence we can thus assume.
\end{proof}

We point out that for us it suffices that for too short tapes, we can inflict an independent penalty on the head state $q_?$ in \cref{lem:success-signalling-TM}.
Whatever happens after that (since the tape is left in superposition) we do not care about, as we will discuss in the next section.
So, as in the 2D case, we do not need to ensure that the QTM behaves properly in this case.

\section{Combining the Marker with the Quantum Turing Machine}\label{sec:combining}
We know how to translate the QPE QTM dovetailed with the universal classical TM from the last section---denoted $M$---into a local history state Hamiltonian $\HTM=\HTM(M,\ivar)$; see \cref{sec:hist-state-intro} and \cite[Thm.~33]{Cubitt2015_long}.
(For brevity, we will refer to this dovetailed QPE QTM and universcal classical TM as the ``universal QTM'' $M$.)
We also assume that we have the Marker Hamiltonian $\Hel$ from \cref{th:elastic!} with an asymptotic falloff exponent $f$ to be specified in due course.

\begin{lemma}\label{lem:HTM}
	Let $\op h$ be the local terms of $\Hel$, and $\op q$ be the local terms of $\HTM$.
	Then on the combined local Hilbert space $\Hs = (\Hs_\mathrm{el}\otimes\Hs_\mathrm{q})^{\otimes N}$, where $N$ is the length of the spin chain, we can define the local Hamiltonian
	\[
		\op h_\mathrm{tot} \coloneqq  \ketbra{\bd}^\perp\otimes\op q + \op h \otimes \1.
	\]
	Then there exists a Hamiltonian $\op H_\mathrm{init}$, such that $\op H \coloneqq  \op H_\mathrm{init} + \sum_i \op h_{\mathrm{tot},i}$ has the following properties:
	\begin{enumerate}
		\item $\op H = \bigoplus_s \op H_s$ block-decomposes like $\Hel$.
		\item All blocks $\op H_s$ of signature $s$, where $s=0$ or $\Hel_s\ge 1$ in \cref{th:elastic!}, have energy $\ge 0$.
		\item On a block of signature $s$ not covered by the previous case $s$ has consecutive 1-bounded segments of length $w_i$.
		\item On a single segment $w_i$, the ground state of $\op H_s$ in $\Hs_\mathrm{q}$ is given by the QTM history state on a tape of length $w_i$,
		$$
		\ket\Psi = \sum_{t=0}^T\ket t\ket{\psi_t}.
		$$
		$\ket{\psi_0}$ is correctly initialized.
		Furthermore, for some $T_1$, $\ket{\psi_{T_1}}$ has overlap $\ge \Omega(2^{-2|\ivar|})$ with a head state $q_?$ from \cref{lem:success-signalling-TM} over a tape qubit in state $\ket-$ on the quantum tape if and only if $w_i<|\phi|+3$ (i.e.\ when the phase estimation truncated).
	\end{enumerate}
\end{lemma}
\begin{proof}
	The first two claims are obvious, since the $\op q$ are positive semi-definite, and the two terms in $\op h_\mathrm{tot}$ commute.
	The third claim follows from \cref{rem:elastic-1,lem:elastic-2}.
	The last claim is the same argument as in the proof in \cite[Thm.~33 and Lem.~51]{Cubitt2015_long}, and the overlap follows from \cref{lem:success-signalling-TM,cor:phase-signal}.
\end{proof}

\subsection{Energy Penalty for Not Halting}
In contrast to the 2D undecidability result, we give an energy penalty to the universal QTM \emph{not} halting.
Since the universal QTM contains a universal TM after the QPE, we have to worry about the case that the universal TM enters a looping state, and runs forever.
Note that by Rice's theorem, we cannot easily exclude this case from all possible inputs that the QPE expands, as deciding whether or not a TM loops is already undecidable.
Luckily this is not an issue in our case, as the following remark shows.
\begin{remark}\label{rem:loop}
If the universal TM enters loops forever, the history state Hamiltonian implementing it will eventually enter a state that can be penalized with a local term.
\end{remark}
\begin{proof}
The way the evolution of the universal TM is encoded in a history state is by performing one computational step every time a counter is incremented.
This counter is itself a classical TM, which is guaranteed to never cycle.
One can therefore easily detect when the counter runs out of space (see sec.\ \cite[4.4]{Cubitt2015_long}), which is when the TM head runs into the right boundary $\bd$ in a state that indicates the incrementing is not terminated yet.
For a base-$\zeta$ counter, this will happen after $\zeta^w$ steps.
\end{proof}
For our purposes a cycling UTM is thus equivalent to one running out of space.

A two-local projector suffices to penalize the QTM head symbol to the left of a boundary marker $\bd$.
We furthermore give a penalty to the head $q_?$ over a $\ket -$ on the quantum tape in \cref{lem:success-signalling-TM} indicating that the phase estimation truncated the expansion prematurely.
We denote the local Hamiltonian term inflicting these penalties with $\op P=\sum_i \ketbra{h_i\bd} +\ketbra{q_?;-}$, where $\{ h_i\}$ is the set of head states we wish to penalize next to the boundary, i.e.\ all QTM states, and the clock TM states indicating that the increment step is not finished yet.

\begin{theorem}\label{th:single-segment}
Let $s=(1, 0, \ldots, 0, 1)$ be a signature of length $w$, and take $\op H_s^{(f)}$ from \cref{rem:elastic!} with a bonus falloff exponent $f$, the universal QTM Hamiltonian $\HTM(\ivar)$ from \cref{lem:HTM}, and the halting penalty term $\op P$.
Further define $\mu=2^{-2|\ivar|}$.
We write $\op H_s=\mu\op H_s^{(f)}+\op H_C$, where $\op H_C \coloneqq \HTM(\ivar) + \op P$ is the circuit Hamiltonian plus non-halting penalty (consistent with \cref{sec:hist-state-intro}).
Then either
\begin{enumerate}
\item $w<|\phi|+5$, i.e.\ the phase estimation truncates the input~\footnote{Truncation happens on less than $|\phi|+3$ tape, as explained in \cref{th:phase-QTM}; here we include the two boundary markers, hence $w<|\phi|+5$.}.
  Then the minimum eigenvalue of $\op H_s$ satisfies $\lmin(\op H_s)>0$, and is strictly monotonically decreasing as $w$ increases.
\item $w\ge|\phi|+5$, i.e.\ the phase estimation finishes exactly, and the universal TM does \emph{not} terminate within the space given. Then, as in the first case, $\lmin(\op H_s) \rightarrow 0$ from above as $w$ grows.
\item $w\ge|\phi|+5$, and the UTM \emph{does} halt after consuming $w_\mathrm{halt}<w$ tape. Then the ground state energy $\lmin(\op H_s)<-\Omega(1/4^{f(w_\mathrm{halt})})$, which in particular is independent of $w$.
\end{enumerate}
\end{theorem}
\begin{proof}
We first note that a history state Hamiltonian encoding a computation of length $T$ that picks up at least one energy penalty, has ground state energy $\lmin\in\Theta(1/T^2)$---see~\cite{Bausch2016a}.
A safe asymptotic lower bound $\bar\lambda_\mathrm{min}<\lmin$ is thus given by $\bar\lambda_\mathrm{min}\coloneqq 1/T^3$.

Furthermore, the runtime of the TM $T$ on the limited space will depend on the available tape space $w$, and on the potential halting time $T_\mathrm{halt}$.
We thus write $T=T(w)$ indicating that the runtime $T$ will be bounded by the tape in the case that the TM cannot terminate within the available space (if it terminates at all).
A trivial runtime bound for $T(w)$ can be derived from Poincaré recurrence.
Since we demand that the TM be reversible, no two configurations of tape and TM head ever repeat.
For $Q$ internal symbols, and $A$ symbols on the tape of length $w$ (where both $Q$ and $A$ are constant), we obtain
\begin{equation}\label{eq:TM-runtime-bound}
T(w)< Q \times w \times A^w=:T_{\max}(w)
\end{equation}
i.e.\ the product of internal states times the possible head positions times all possible tape configurations.
\Cref{eq:TM-runtime-bound} allows us to choose a falloff exponent $f$ such that
\begin{equation}
\frac{1}{\bar\lambda_\mathrm{min}}=T^3(w)<T_{\max}^3(w)<2^{f(w)},
\label{eq:falloff-1}
\end{equation}
e.g.\ $f(w)=2^w$ for a choice of $d=7$ in \cref{rem:elastic!}.

We can lower bound the ground state energy of the history state Hamiltonian plus penalty part of $\op H_s$, i.e.\ $\HTM(\ivar) + \op P$ \emph{without} the energy bonus inflicted within $\op H_s^{(f)}$, in relation to the segment length $w$ by $\bar\lambda_\mathrm{min}(T)=\bar\lambda_\mathrm{min}(T(w))$ as shown in \cref{fig:segment-energy}.
\begin{figure*}
\includegraphics{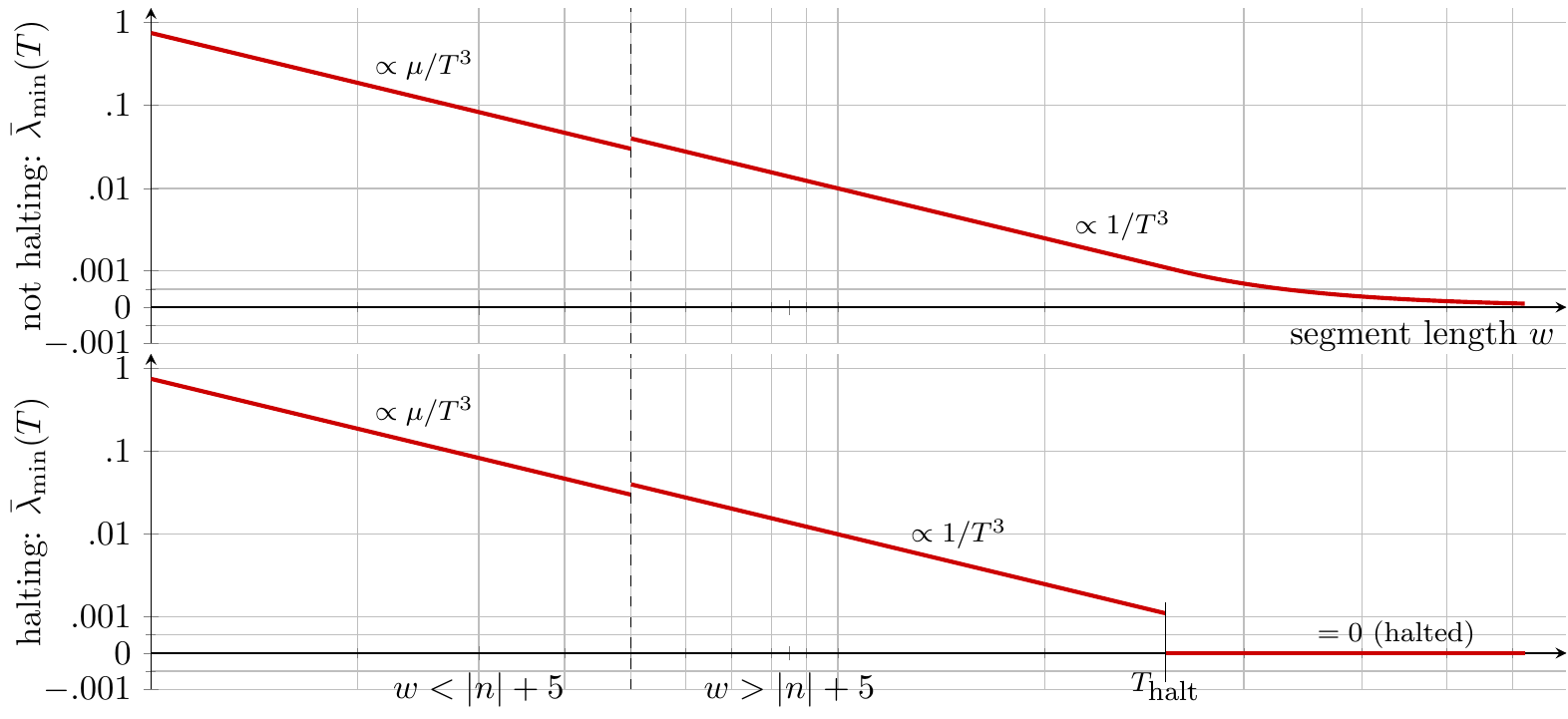}
\caption{\label{fig:segment-energy}Lower bounds to the groundstate energy of the QTM history state Hamiltonian on a single segment, as a function of the segment length, in the halting and non-halting cases.}
\end{figure*}
The top panel shows the case for which the dovetailed universal TM will not halt.
Depending on the segment length $w$, we have the following two cases:
\begin{enumerate}
\item For $w<|\ivar|+5$, there is not enough tape for the phase expansion.
By \cref{lem:success-signalling-TM}, we know that with probability $\ge\mu$, the phase estimation results in a string where the head symbol $q_?$ is over a tape qubit $\ket-$, which shows the phase estimation truncated the output.
Therefore, the head will be penalized by $\op P$ with overlap $\ge\mu$.
In order to account for the fact that the part of the computation following on from the garbage state coming out of the interrupted phase estimation could well halt, even if $\phi$ encodes a non-halting instance, we scale the lower bound in this area down by a factor $\mu$---which is still non-negative, as $\mu$ is just a constant prefactor.
Observe that it is not essential that we inflict the penalty term at the end of the history state (see e.g.\ \cite[Cor.~44]{Bausch2016}).
\item For $w\ge|\ivar|+5$
  the phase estimation finishes exactly, and the universal TM retrieves the complete input on which it will not halt; the energy penalty $\op P$ applies as well.
\end{enumerate}
In either case, the history state evolution is of length $T=T(w)$, i.e.\ the runtime of the computation until the head bumps into the right marker or the clock \emph{driving} the computation runs out of time, both of which depends on the segment length $w$.
In both cases, the last step of the computation will be completely penalized.
This pushes the corresponding associated Hamiltonian's ground state energy up by $\Theta(1/T^2)$.

In case the dovetailed universal TM \emph{does} halt, there is no further forward transition~\footnote{In the 2D result, the Turing machine then entered a time wasting operation, where the head would simply idle until the clock runs out of time. This is, strictly speaking, not necessary: the history state evolution can stop at any point, while keeping the computation reversible---see e.g.~\cite{Bausch2016}. If in doubt, it is of course always possible to use the traditional way such that once the clock runs out of space, the penalty is only inflicted if the TM is not yet in a halting configuration.}.
The TM head will not feel the penalty $\op P$, and the ground state energy is that of an unfrustrated history state Hamiltonian, i.e.\ zero.
Observe that this happens at a point $T_\mathrm{halt}$ which is obviously independent from $w$.
The precise statement is that once there is enough tape such that the entire evolution of the (halting) TM can be contained, no halting penalty will be felt.
This happens once $w$ is such that $T(w) \ge T_\mathrm{halt}$.
Define this segment length to be $w_\mathrm{halt}$.

After including the Marker Hamiltonian $\op H_s^{(f)}$ in $\op H_s$, we obtain the ground state energy bounds shown in \cref{fig:pen-and-bonux}.
The dashed blue line shows an upper bound on the negative magnitude of the energy bonus $E(w)$ induced by the Marker Hamiltonian $\mu\op H_s^{(f)}$ with $f(w)=2^w$.
Note that we chose the loose bound $-\mu/T^{3+\delta}$ for visualization purposes.
By \cref{rem:elastic!}, we know that this bonus in fact satisfies $-\mu/2^{f(w)} \le E(w)$.
With \cref{eq:falloff-1}, we know that
\begin{align}
T^3(w) &< 2^{f(w)} %\nonumber\\
\Longleftrightarrow \frac{1}{T^3(w)} > \frac{1}{2^{f(w)}} \ge -E(w),\nonumber\\
\intertext{and thus clearly}
\lmin + E(w) &\ge \bar\lambda_\mathrm{min} + E(w) > 0.
\label{eq:with-bonus-still-positive}
\end{align}

Observe that $\op H_s^{(f)}$ commutes with both $\HTM(\ivar)$ and $\op P$, so the resulting ground state energy of $\op H$ for the block of segment length $w$ will simply be
\[
\lmin(\op H_s)=E(w) + \lmin(\HTM(\ivar) + \op P).
\]
The solid red line shows the lower bound achieved by subtracting the smaller attractive contribution $E(w)$ from the lower bound for $\lmin(\HTM(\ivar) + \op P)\ge\bar\lambda_\mathrm{min}$.
We again consider each case separately.

If the dovetailed UTM does not halt, we subtract $|E(w)|$ (or, equivalently, add $E(w)$, since $E(w)$ is negative) from the lower bound we proved before.
The ground state energy $\lmin+E(w)>0$ by \cref{eq:with-bonus-still-positive}.

If the UTM \emph{does} halt, on the other hand, there exists a halting time $T_\mathrm{halt}$ such that $\bar\lambda_\mathrm{min}(w)=0$ for all $w>w_\mathrm{halt}$ (see magnified area).
This immediately implies that
\[
\lmin(w)
\begin{cases}
\ge \bar\lambda_\mathrm{min}(w)+E(w)>0 & \text{for all $w<w_\mathrm{halt}$} \\
=-f(w)<0 & \text{otherwise.}
\end{cases}
\]
This proves the last claim.
\end{proof}

We observe that in the halting case, the energy is smallest when the segment length is precisely $w_\mathrm{halt}$, as $|E(w)|$ is strictly monotonically decreasing.

In light of \cref{rem:elastic!}, i.e.\ the fact that $\Hel$ breaks into signature blocks $\Hel_s$, we want to extend \cref{th:single-segment} to the case where the signature $s$ is not just a single segment, but a series of segments of varying length.
We capture this in the following lemma.

\begin{lemma}\label{lem:multiple-segments}
Let the notation be as in \cref{th:single-segment}, but take a signature $s$ with potentially multiple segment lengths $(w_i)_i$ as in \cref{th:elastic!}.
Let $\nu(w_i)$ be the energy of the ground state of a block segment of length $w_i$.
Then
\[
    \lmin(\op H_s)=\sum_i \nu(w_i).
\]
\end{lemma}
\begin{proof}
Both $\HTM(\ivar)$ and $\op P$ commute with $\op H_s^{(f)}$, and we use the same Cartesian graph product argument for the latter as in \cref{lem:elastic-3}.
\end{proof}

This leads us to the main technical theorem.
\begin{theorem}\label{th:TM-ham}
For any Turing machine $M$ and input $\ivar\in\field N$ to $M$, we can explicitly construct a sequence of 1D, translationally invariant, nearest-neighbour Hamiltonians $\op H_N(\ivar, M)$ on the Hilbert space $(\field C^d)^{\otimes N}$ with the property that either
\begin{enumerate}
\item $M(\ivar)$ does not halt, and $\lmin(\op H_N)\ge0$ for all $N$, or
\item $M(\ivar)$ halts, and
\[
\lmin(\op H_N) \begin{cases}
<-\lfloor N / w_\mathrm{halt}\rfloor \Omega(1/T_\mathrm{halt}^{3}) & \text{$N>w_\mathrm{halt}$} \\
\ge 0 & \text{$N\le w_\mathrm{halt}$,}
\end{cases}
\]
where $T_\mathrm{halt}$ is the time needed for $M(\ivar)$ to halt, and $w_\mathrm{halt}$ is the length of the tape accessed during the computation.
\end{enumerate}
\end{theorem}
\begin{proof}
We set $\op H\coloneqq \op H_C(M,\ivar) + \Hel$ for $\op H_C(M,\ivar) = \HTM(M,\ivar)+\op P$, and with $f(w)=w^2$ as in \cref{lem:multiple-segments}, but with the full Marker Hamiltonian $\Hel$ instead of a single signature block.
We already know that $\op H$ is block diagonal, and by \cref{lem:multiple-segments} we know the spectrum of each block.
There are two cases.
\begin{enumerate}
\item $M(\ivar)$ does not halt.
By \cref{th:single-segment}, we know that the ground state energy contribution of a single segment is falling off monotonically with the segment length.
By \cref{lem:multiple-segments}, we know that the overall ground state energy is the sum of the individual segments.
The block with the lowest energy is thus the one with a single segment of length $N$, and in particular non-negative (or if we do not penalize the rightmost halting boundary then the ground state energy is zero).
\item $M(\ivar)$ halts after $T_\mathrm{halt}$ steps, having consumed $w_\mathrm{halt}$ tape.
If $N<w_\mathrm{halt}$ the same argument as above holds.
If $N>w_\mathrm{halt}$, we have space for at least $\lfloor N/w_\mathrm{halt} \rfloor$ segments of tape on which the TM terminates.
It is beneficial to have as many such segments as possible, as each of these contributes an energy $E(w_\mathrm{halt})<0$.
Ignoring the right-most segment of non-full length (which is a single constant energy penalty), the block with a signature where the shortest possible segments on which the TM can halt are left-aligned has the lowest energy $<-\Omega(1/T^3_\mathrm{halt})$.
Since there is only a single rightmost segment, but $O(N)$ bonus'ed segments, the asymptotic bound is $\lmin(\op H_N)<-\Omega(1/T_\mathrm{halt}^3)$.
\end{enumerate}
The claim follows.
\end{proof}
Note that the ground state energy of $\op H$ diverges to minus infinity in the halting case, but the ground state energy \emph{density} is bounded.

\section{Undecidability of the Spectral Gap}\label{sec:undecidability}
In order to obtain the full result, we will need to shift the energy spectrum of $\op H$ from \cref{th:TM-ham} up so that its ground state is either $\ge 1$, or diverges towards $-\infty$, add a trivial Hamiltonian with ground state energy $0$, and another Hamiltonian with continuous spectrum.
We begin by observing that an energy shift is readily achieved as follows.
\begin{lemma}\label{lem:shift-ham}
By adding at most two-local identity terms, we can shift the energy of $\op H$ from \cref{th:TM-ham} such that
\[
    \lmin(\op H) \begin{cases}
        \ge 1 & \text{if the TM does not halt,} \\
        \longrightarrow-\infty & \text{in the halting case.}
    \end{cases}
\]
\end{lemma}
\begin{proof}
Employ \textcite{gottesman2009quantum}'s boundary trick again (cf.\ \cref{rem:elastic-1}), which hinges on the fact that there is $N$ one-local but only $N-1$ two-local terms.
\end{proof}

The next step is to construct a simple Hamiltonian with a unique ground state of energy $0$, and a spectral gap of $1$.
\begin{lemma}\label{lem:trivial-ham}
There exists a one-local translationally-invariant Hamiltonian $\Htriv$ on $(\field C^2)^{\otimes N}$ which is diagonal in the computational basis, with unique zero-energy ground state $\ket{00\cdots0}$, and all other $\lambda\in\spec(\Htriv)$ satisfy $\lambda\ge1$.
\end{lemma}
\begin{proof}
Take $\Htriv = \sum_{i=1}^N \ketbra 1_i$.
\end{proof}

Furthermore, we need a Hamiltonian with continuous spectrum in $[0,\infty)$ in the thermodynamic limit, which we call $\Hcont$.
%There are many choices for $\Hcont$ in the literature, e.g.\ the 1D critical XY model~\cite{Lieb1961a}.
With this, we can prove the following theorem.
\newcommand{\Htot}{\op H_\mathrm{tot}}
\begin{theorem}\label{th:undecidability-1-gap}
  Take $\op H$ from \cref{th:TM-ham} with shifted energy as in \cref{lem:shift-ham}, and let $\Hs_C$ denote the Hilbert space on which it acts.
  Take $\Hcont$ as defined and denote the Hilbert space on which it acts $\Hs_\mathrm{dense}$.
  Finally, let $\Htriv$ be the trivial ground-state-energy $0$ Hamiltonian from \cref{lem:trivial-ham} with Hilbert space $\Hs_\mathrm{trivial}$.
  Then we can construct a Hamiltonian $\Htot=\Htot(\op H,\Hcont,\Htriv)$ on $\Hs\coloneqq (\Hs_C\otimes\Hs_\mathrm{dense})\oplus\Hs_\mathrm{trivial}$ as in \cref{sec:hist-state-intro} such that
\[
    \spec(\Htot) = \{0\} \cup (\spec(\op H) + \spec(\Hcont)) \cup G,
\]
where $G\subset[1,\infty)$.
\end{theorem}
\begin{proof}
We use a trick from~\cite{Bausch2015}.
Define
\[
    \Hguard\coloneqq \sum_{i=1}^N(\1_{1,2}^{(i)}\otimes\1_3^{(i+1)} + \1_3^{(i)}\otimes\1_{1,2}^{(i+1)}).
\]
It is clear that any state with support on both $\Hs_C\otimes\Hs_\mathrm{dense}$ and $\Hs_\mathrm{trivial}$ will incur an energy penalty from $\Hguard$.
Define further
\[
    \Htot = \op H\otimes\1_2 \oplus 0_3 + \1_1\otimes\Hcont \oplus 0_3 + 0_{1,2} \oplus \Htriv + \Hguard.
\]
Then the claim follows.
\end{proof}

Since the halting problem is undecidable in general, we obtain our main result \cref{th:main}, which we re-state in the following way.
\begin{theorem}[Undecidability of the Spectral Gap in 1D]
	Let $\beta\in(0,1]$ be arbitrary.
    Whether the Hamiltonian in \cref{th:undecidability-1-gap} is gapped with a spectral gap of $1$, or is gapless, is undecidable, even if we multiply $\op H$ and $\Hcont$ by $\beta$.
$\op H_\mathrm{tot}$ can then be assumed to comprise local terms as laid out in \cref{th:main}.
\end{theorem}
\begin{proof}
We note that the properties required from $\op H$ and $\Hcont$ in \cref{th:undecidability-1-gap} remain true, independent of any constant prefactor $\beta$; i.e.\ the spectral gap for
\begin{align*}
    \Htot =&\ \beta(\op H\otimes\1_2 \oplus 0_3 + \1_1\otimes\Hcont \oplus 0_3) \\
    +&\ 0_{1,2} \oplus \Htriv + \Hguard.
\end{align*}
remains undecidable, for all $\beta>0$.

In addition, this means we can assume wlog that the local terms of $\op H$ and $\Hcont$ have norm $\|\cdot\|\le1$ for $\beta \le 1$. The estimates of the norms in \cref{th:main} then stem from computing the norms of the terms in $\Htriv$ and $\Hguard$.
\end{proof}

\section{Extensions of the result}\label{sec:extensions}
\subsection{\label{sec:periodic}Periodic Boundary Conditions}
\Cref{th:main} can, in a limited fashion, be extended to periodic boundary conditions, which we summarize in the following lemma and theorem.
\begin{lemma}\label{th:main-periodic}
\Cref{th:main} holds, even on 1D spin chains with periodic boundary conditions, and under the assumption that the spin chain instances all have length coprime to $P$, at the cost of a local dimension that grows with $P$.
\end{lemma}
\begin{proof}
Take the Hamiltonian from \Cref{th:main}.
The only difference to the open boundary conditions case is that there is no mismatch between the number of 1- and 2-local terms, so we will have to modify those parts of the proof carefully.

We first note that \cref{rem:elastic-1} relies on this boundary trick.
In the periodic case, however, we cannot use it.
The reason for \cref{rem:elastic-1} was to enforce all segments to have right-boundaries---otherwise a segment which is half-unbounded on the right would pick up the bonus from the marker Hamiltonian, but no penalty due to the TM running out of tape.
This problem never occurs on a ring: if there is at least one marker present, it is automatically guaranteed that each segment is properly bounded.
Therefore, if we drop the term $\op P'$, \cref{lem:elastic-2} goes through, but such that the resulting Hamiltonian has a ground state energy of $0$, not $-4$.

The next step which needs amendment is in \cref{th:elastic}, where we note that there is no leftover penalty of $1/2$ from the leftmost boundary marker---bonus and penalty terms from \cref{lem:elastic-3} precisely cancel.
To this end, there is no energy shift necessary.

The last issue is with \cref{lem:trivial-ham}:
while one can straightforwardly create a Hamiltonian with constant negative ground state energy when there are open boundary conditions, this is not the case with periodic systems.
To circumvent this, we assume we have a trivial Hamiltonian $\Htriv$ with unique classical ground state with energy 0 and first excited state 1.
We then shift \emph{everything else} up by a constant.
Under the stated assumption that the spin loop has a length coprime to $P$, the positive energy shift can be achieved by adding an ancilliary Hilbert space of dimension $P$, and adding local projectors that enforce a tiling \'a la $1, 2, 3, \ldots, P$.
Since this tiling has to be broken at least at one site on the ring, there is a constant energy shift.

The overall Hamiltonian then reads, as before,
\[
\op H_\text{tot}=\op H'\otimes\1_2\oplus 0_3 + \1_1\otimes\Hcont\oplus 0_3 + 0_{1,2}\oplus\Htriv + \Hguard.
\]
where $\op H'$ equals $\op H$ from \cref{th:main}, with the $P$-periodic tiling enforced.
In the non-halting case, $\op H_\text{tot}$ will be gapped with $\Delta\ge1$, and unique ground state.
In the halting case, $\op H'$ will have an energy that diverges to $-\infty$ (despite the constant energy shift inflicted by the $P$-periodic tiling), and therefore pulls the dense spectrum of $\Hcont$ with it.
The claim of the theorem follows.
\end{proof}

\subsection{\label{sec:transverse-only}Purely Transverse Field $\ivar$ Dependence}
Thus far, the terms in \cref{th:main} explicitly-dependent on the phase $\phi$ are two-local.
More specifically, there are the one-local terms $\op a'$ with a coefficient of $\beta 2^{-2|\ivar|}$, as well as the terms $\op b'''$ with prefactors $\exp(\pm\ii \pi \phi(\ivar))$ and $\op b''''$ with prefactors $\exp(\pm \ii \pi 2^{-2|\ivar|})$, respectively.

We can strengthen our findings by making the $\ivar$-dependent terms all one-local.
This is a straightforward observation, and we will leave the details to the reader.
\begin{remark}
There exists a variant of the QPE QTM such that the corresponding Hamiltonian $\HTM(\ivar)$ has only one-local terms that depend on $\ivar$.
\end{remark}
\begin{proof}
The two-local terms dependent on $\ivar$ stem from two steps of the phase estimation algorithm:
\begin{enumerate}
\item The controlled-phase gates with powers of the gate $\op U_\phi$, and
\item the inverse QFT with powers of the controlled rotation $\op U_\alpha$.
\end{enumerate}
Naturally, any modification to the QPE QTM will directly translate to the corresponding history state Hamiltonian $\HTM$; in particular, if we manage to modify the algorithm to make the gates that depend on $\ivar$ one-local, the resulting Hamiltonian can be rendered one-local as well.
To this end, we first note the circuit identity
\[
\mbox{
    \Qcircuit @C=2em @!R=0em  {
        & \ctrl{1} & \qw \\
        & \gate{\op U} & \qw
    }}
    \hspace{2.1em}
    \raisebox{-1.45em}{=}
    \hspace{2em}
\mbox{
    \Qcircuit @C=2em @!R=0em {
        & \gate{\op V} & \ctrl{1} & \qw & \ctrl{1} & \qw \\
        & \gate{\op V} & \targ & \gate{\op V^\dagger} & \targ & \qw
    }
  }
\]
where $\op V = \sqrt{\op U}$, as e.g.\ explained in \cite[fig.\ 4.6]{Nielsen2010}.
Furthermore, we note that a generic translation of a circuit gate $\op V$ to a Hamiltonian---say at time step $t$---results in a local term \`a la
\[
    \op h_{\op V,t} = \sum_i \big( \ket{t}\ket{i} - \ket{t+1}\op V\ket i \big)\big( \bra{t}\bra{i} - \bra{t+1}\bra i\op V^\dagger \big).
\]
The locality of $\op h_{\op V,t}$ thus crucially depends on how the clock is implemented, and there exists a long history of development rendering those transitions two-local \cite{Bausch2016,gottesman2009quantum}.
Yet for our purposes we would like said gate to be implemented exactly, and such that if $\op V$ depends on $\ivar$, the overall term does not become two-local.
This can be achieved by ensuring that the clock transition $\ket t \longmapsto \ket{t+1}$ is a geometrically one-local term on the physical spins, during which $\op V$ is applied to the quantum register contained within the very same spin.
It is clear that this can be done in a translationally-invariant fashion within the context of Feynman's standard circuit-to-Hamiltonian construction, at the cost of increasing the local dimension slightly.
\end{proof}

\section{Discussion}

In spite of indications that 1D spin chains are simpler systems than higher dimensional lattice models, we have shown that the spectral gap problem is undecidable even in dimension one, settling one of the big open questions left in~\cite{Cubitt2015}. At the same time, the construction we present has some distinguishing features from the 2D construction.

In the 2D case, the ground state behaves as a highly non-classical model, showing all features of criticality, for any system size where the Universal Turing machine embedded in the model does not halt. If the machine eventually halts, starting from the corresponding system size the ground state will abruptly transition to a classical, product state. The construction we have presented shows the \emph{opposite} property: the ground state is a product classical state of the trivial Hamiltonian $\sum_{i=1}^N \ketbra 1_i$ (i.e.\ $\ket{00\cdots0}$), unless the machine halts, in which case the low-energy spectrum of the Hamiltonian suddenly begins to converge to a dense set.

While both the 1D and 2D cases can be seen as an example of a
size-driven phase transition~\cite{Bausch2015}, in the 1D construction we transition from trivial to gapless, instead of vice-versa.
In 1D, there are algorithms with \emph{provably} polynomial running time in the system size~\cite{Landau2013}.
Nonetheless, our results prove that any numerical study of the ground-state properties will \emph{not} reveal any of the phenomena one would expect of a gapless system. For both gapless and gapped cases, the numerics will instead find a classical ground state with constant gap above it, all the way up to some threshold chain length, which is uncomputable in general (determined by the tape length required for the universal Turing machine to halt, if it does indeed halt).

Therefore, not only is there no algorithm that can correctly predict whether a 1D Hamiltonian is gapped or not, but also the known efficient algorithms for computing ground state properties will fail to predict the correct thermodynamic properties of the state---even properties as elementary as the decay of correlations.

Of course the polynomial runtime of ground state approximation algorithms hinges on the promise that the one-dimensional system under study is asymptotically gapped, which is why we do not expect them to work e.g.\ for Hamiltonians with a QMA hard ground state problem.
In fact, because of this, we know that all Hamiltonians of one-dimensional spin chains with a QMA hard ground state problem have to be gapless.
Consequently, our 1D result implies that the premise itself on which all the efficient algorithms rely turns out to be undecidable.

Our findings extend to periodic boundary conditions, albeit in a limited fashion, for a number of spins promised to be coprime to some number $P$ (Section \ref{sec:periodic}).
This comes at the cost of a local dimension that grows linearly with $P$.
The general periodic case with fixed local dimension remains open.
We further showed that the same result holds for the case where the $\ivar$-dependence is only on the one-local transverse field, and all two-local terms are fixed (Section \ref{sec:transverse-only}).
As in 2D, the reduction also demonstrates that the ground state energy density of 1D spin chains is, in general, uncomputable.

% Discussion on genericality and disorder

An important question to ask is whether undecidability of the spectral gap is a generic feature, e.g.\ in the measure-theoretic sense over some underlying distribution, and whether the construction obeys some form of stability with respect to perturbations.
The strongest-known stability proofs for general local perturbations only apply to certain types of frustration-free Hamiltonian\cite{BravyiHastingsMichalakis,Michalakis,Yarotsky}. Little is known about stability of the spectral gap to arbitrary local perturbations even for much simpler and far better--studied models than ours, such as the 1d transverse Ising model.

% We present a more extensive discussion in \cref{sec:stability}, but summarize the main points in the following.
We can say a little more in our specific case.
First, note that generic \emph{disordered} local Hamiltonians---i.e.\ those where the local interaction terms are chosen uniformly at random, in particular not translationally invariant---are known to be gapless in 1D \cite{Movassagh2017}. This no longer holds true in the case of frustration-free translation invariant interactions \cite{Lemm2019}, while the question for generic translational invariant models is still open.
Whether or not a random instance sampled from a collection of local terms is generically gapped or gapless thus depends on the choice of the underlying distribution. If, for instance, the local terms are those given in \cref{th:main} (i.e.\ such that the phase gate for $\ivar$ and the encoded binary expansion length $|\ivar|$ match), and our choice of random distribution is over the encoded inputs $\ivar$---sampled with regards to some powerlaw distribution, say---then the probability that a random sample thereof has a decidable spectral gap depends on the universal Turing machine used within the construction.
Yet in general even this probability itself can be uncomputable (cf.\ Chaitin's constants \cite{Barmpalias2012}).

Similarly, the type of perturbation we allow determines whether we expect any type of robustness properties to hold.
For instance, varying the classical couplings in \cref{th:main} (e.g.\ $\op a$ perturbed by a term $\epsilon\bar{\op a}$ for small $\epsilon$) leaves the construction intact: any such term encoding a penalty or bonus diagonal in the computational basis will only change the corresponding energies by an amount $\propto\epsilon$.
But bonus and penalty terms are of order one: for sufficiently-small $\epsilon$, the limiting behaviour of $\op H_N$ remains unaltered.
In contrast, the Feynman-Kitaev Hamiltonian terms are not robust to perturbations \cite{Bausch2016a}---any such perturbation tends to produce a localized ground state which we expect to break the intended behaviour of our construction.

Perturbing the parameter $\ivar$ appearing in \cref{eq:mainthm-two-sites-interaction}  will simply change the instance simulated by the UTM; yet for a generic such perturbation, the binary expansion $\phi(\ivar)$ is of course infinitely long. Since our construction cannot provide all those gates with bounded-norm local terms, this type of perturbation cannot be analysed within the scope of Hamiltonians we construct.
On the other hand, perturbing the phase term $\phi(\ivar)$ such that its binary expansion remains bounded simply changes the encoded Turing machine input. This could, of course, change the behaviour between halting and non-halting.
It is therefore intrinsic to this construction---and expected for any undecidable property of the Hamiltonian---that no form of stability should hold in the encoded phase.

% comment on entanglement

An important point to emphasize is that, in the halting case, the critical behaviour exhibited by $\op H_N$ depends on the behaviour of the dense spectrum Hamiltonian $\op H_{\mathrm{dense}}$.
As proven in \cref{th:TM-ham}, the ground state of $\op H_C$ in the halting case is a product of segments of length just long enough for the encoded Turing machine to halt. While each segment can \emph{individually} be uncomputably large (e.g.\ if the instance was a Busy Beaver) with a correspondingly uncomputable amount of entanglement entropy, the \emph{overall} ground state remains product across these individual segments. This implies that even in the gapless phase, the entanglement entropy of the ground state of $\op H_C$ is independent of the system size $N$, and only depends on the parameter $\ivar$.

Therefore, if the ground state of $\op H_{\mathrm{dense}}$  has large entanglement entropy, so will have the ground state of $\op H_N$, and thus detecting violation of entanglement entropy area laws is also undecidable. On the other hand there are instances of 1D Hamiltonians with dense spectrum whose ground states do not have large entanglement entropy \cite{Wolf2006,FernndezGonzlez2014}. Choosing such Hamiltonians in the construction will give a family of Hamiltonians $\op H_N(M,\ivar)$ that always obeys an entanglement area law. Thus criticality is not essential to undecidability of the spectral gap; undecidability is possible even in cases where critical behaviour is guaranteed not to occur.

We conclude by mentioning an open question which is still to be addressed.
As in the case of 2D systems, the model we present is extremely artificial, with a very large local dimension,
which we  did not try to optimize.
It is an interesting problem whether it is possible to find more natural models exhibiting undecidable properties, or whether there is a local dimension threshold below which quantum systems necessarily behave in a predictable way~\cite{Bravyi2015}. I.e.\ does it hold that below some threshold on the local dimension the spectral gap problem becomes decidable?
While size-driven phase transitions can happen in 2D with very small local dimension~\cite{Bausch2015}, these low-dimensional constructions are decidable.
Determining if this threshold exists and if and how it depends on the lattice dimension remains a very interesting open question. The only known result in the other direction, proving decidability for frustration-free, nearest-neighbour qubit chains \cite{Bravyi2015}, is also specific to 1D. Together with our 1D undecidability result, this gives strong evidence that the dimension threshold has a non-trivial answer.

\begin{acknowledgements}
J.\,B.\ acknowledges support from the German National Academic Foundation, the EPSRC (grant 1600123), and the Draper's Research Fellowship at Pembroke College.
T.\,S.\,C.\ is supported by the Royal Society.
A.\,L.\ acknowledges support from the European Research Council (ERC Grant Agreement no. 337603) and VILLUM FONDEN via the QMATH Centre of Excellence (Grant no. 10059), the  Walter  Burke  Institute  for Theoretical Physics in the form of the Sherman Fairchild Fellowship as well as support from the Institute for Quantum  Information  and  Matter  (IQIM),  an  NSF  Physics Frontiers Center (NFS Grant PHY-1733907).
D.\,P.\,G.\ acknowledges financial support from Spanish MINECO (grants MTM2014-54240-P, MTM2017-88385-P and Severo Ochoa project SEV-2015-556), Comunidad de Madrid (grant QUITEMAD+CM, ref. S2013/ICE-2801) and the European Research Council (ERC) under the European Union's Horizon 2020 research and innovation programme (grant agreement No 648913).
\end{acknowledgements}

\bibliography{bibliography}
\end{document}